\documentclass[journal]{IEEEtran}%

\usepackage{graphicx}
\usepackage{epsfig}
\usepackage{times}
\usepackage{amsmath}
\usepackage{amsthm}
\usepackage{thmtools,thm-restate}
\usepackage{amssymb}
\usepackage[section]{placeins} 
\usepackage{placeins}
\usepackage{amsmath} 
\usepackage{multirow}
\usepackage{graphicx}
\usepackage[normalem]{ulem}
\usepackage{subcaption}
\usepackage{algorithm,tabularx}
\usepackage{algpseudocode}
\usepackage{amsfonts}
\usepackage{cases}
\usepackage{romannum}
\usepackage{textcomp}
\usepackage{xcolor}%
\usepackage{textcomp}
\providecommand{\U}[1]{\protect\rule{.1in}{.1in}}
\providecommand{\U}[1]{\protect\rule{.1in}{.1in}}

\IEEEoverridecommandlockouts

\newtheorem{assumption}{Assumption}
\newtheorem{theorem}{Theorem}
\newtheorem{corollary}{Corollary}
\newtheorem{lemma}{Lemma}
\newtheorem{proposition}{Proposition}
\newtheorem{remark}{Remark}
\newtheorem{definition}{Definition}

\useunder{\uline}{\ul}{}
\makeatletter
\newcommand{\multiline}[1]{  \begin{tabularx}{\dimexpr\linewidth-\ALG@thistlm}[t]{@{}X@{}}
#1
\end{tabularx}
}
\makeatother
\usepackage{cite}

\usepackage{picins}
\usepackage{centernot}
\let\labelindent\relax
\usepackage{enumitem}
\setlist[itemize]{leftmargin=*}
\usepackage{makecell}
\usepackage[normalem]{ulem}
\usepackage{cuted}
\usepackage{hyperref}
\usepackage{stfloats}
\usepackage{lscape}

\newcommand{\R}{\mathbb{R}}
\newcommand{\N}{\mathbb{N}}

\newcommand{\T}{\top}
\newcommand{\I}{\mathbf{I}}
\newcommand{\0}{\mathbf{0}}

\newcommand{\F}{\mathcal{F}}

\newcommand{\X}{\mathcal{X}}
\newcommand{\diag}{\text{diag}}
\newcommand{\tsup}[1]{\textsuperscript{#1}}
\newcommand{\mb}[1]{\mathbf{#1}}
\renewcommand{\H}{\mathcal{H}}

\newcommand{\bm}[1]{\begin{bmatrix}#1\end{bmatrix}}

\algdef{SE}[SUBALG]{Indent}{EndIndent}{}{\algorithmicend\ }%
\algtext*{Indent}
\algtext*{EndIndent}

\title{\LARGE \bf
Dissipativity-Based Decentralized Co-Design of Distributed Controllers and Communication Topologies for Vehicular Platoons
}

\author{Shirantha Welikala, Zihao Song, Panos J. Antsaklis and Hai Lin 
\thanks{
The first two authors contributed equally to this paper.

The support of the National Science Foundation (Grant No. CNS-1830335, IIS-2007949) is gratefully acknowledged. The authors are with the Department of Electrical Engineering, College of Engineering, University of Notre Dame, IN 46556, \texttt{{\small \{wwelikal,zsong2,hlin1,pantsakl\}@nd.edu}}.}}

\begin{document}

\maketitle

\pagenumbering{arabic}
\thispagestyle{plain}
\pagestyle{plain}

\begin{abstract}

Vehicular platoons provide an appealing option for future transportation systems. Most of the existing work on platoons separated the design of the controller and its communication topologies. However, it is beneficial to design both the platooning controller and the communication topology simultaneously, i.e., controller and topology co-design, especially in the cases of platoon splitting and merging. 
We are, therefore, motivated to propose a co-design framework for vehicular platoons that maintains both the compositionality of the controller and the string stability of the platoon, which enables the merging and splitting of the vehicles in a platoon.
To this end, we first formulate the co-design problem as a centralized linear matrix inequality (LMI) problem and then decompose it using Sylvester's criterion to obtain a set of smaller decentralized LMI problems that can be solved sequentially at individual vehicles in the platoon.
Moreover, in the formulated decentralized LMI problems, we encode a specifically derived local LMI to enforce the $L_2$ stability of the closed-loop platooning system, further implying the $L_2$ weak string stability of the vehicular platoon. 
Finally, to validate the proposed co-design method and its features in terms of merging/splitting, we provide an extensive collection of simulation results generated from a specifically developed simulation framework\footnote{Available in GitHub: \href{HTTP://github.com/NDzsong2/Longitudinal-Vehicular-Platoon-Simulator.git}{HTTP://github.com/NDzsong2/Longitudinal-Vehicular-Platoon-Simulator.git}.} that we have made publicly available.

\end{abstract}

\section{Introduction}\label{Sec:Introduction}






Vehicular platoon is a promising solution for future transportation development. 
By arranging all vehicles in closely spaced groups, platoons lead to a decreased aerodynamic drag and a better utilization of road infrastructures \cite{besselink2017string}. In this way, it can significantly increase the highway capacity, enhance the traffic safety and reduce fuel consumption and exhaust emissions \cite{jia2015survey}. 


In vehicular platoons, there are many interesting control problems, such as designing controllers to achieve the desired inter-vehicle distances and speeds. 
Plenty of control methods have been proposed, such as linear control (e.g., PID \cite{fiengo2019distributed}, LQR/LQG \cite{tavan2015optimal,wang2022optimal}, $H_{\infty}$ control \cite{ploeg2013controller,herman2014nonzero}) and nonlinear control approaches (e.g., model predictive control (MPC) \cite{chen2018robust}, optimal control \cite{morbidi2013decentralized,ploeg2013controller}, consensus-based control \cite{syed2012coordinated}, sliding-mode control (SMC) \cite{xiao2011practical,guo2016distributed}, backstepping control \cite{zhu2018distributed,chou2019backstepping} and intelligent control \cite{ji2018adaptive,li2010design}).

In most of these existing works, it is usually assumed that the communication topology of the platoon is fixed.
However, in many practical scenarios, the communication topology is dynamic, and vehicles can join or leave a platoon anytime. For example, dynamic changes, such as merging and splitting of the vehicles, frequently occur when ramp vehicles merge into mainlines, when long platoons pass intersections, and when vehicles exit a platoon and make lane changes \cite{guo2015communication}. This poses significant challenges in the controller design for platoons since we must  simultaneously guarantee the designed controller's compositionality and the platoon's string stability after vehicles join/leave the corresponding formations.
On the one hand, the compositionality of the controllers implies that the designed controllers for other remaining vehicles in the platoon are not required to be re-designed, e.g., re-tuning the control parameters from scratch, when such maneuverings occur.
On the other hand, the string stability of the platoons captures the uniform boundedness of the tracking errors as they propagate along the downstream direction of the platoons. 


The study of the platooning control that enables merging and splitting
is only recent. Methods, such as PID \cite{dasgupta2017merging}, MPC \cite{goli2019mpc}, and cooperative approaches \cite{morales2016merging}, have been proposed in the literature, which are either heuristic or traditional formation-based control without the assuring of the string stability. More importantly, re-designing or re-tuning of the controllers are needed for those remaining vehicles after the vehicular joining/leaving the platoon. In other words, these methods are not compositional and lack formal string stability guarantee. 

We are, therefore, motivated to study the platooning control that can simultaneously guarantee the designed controller's compositionality and the platoon's string stability when vehicles join/leave the corresponding formations. Furthermore, it is known that the communication topology of the platoons plays a significant role in the performance of the platoons, e.g., stability margin, communication cost, security, etc. However, the study of topology-related impacts on the platoon is very sparse. Existing work includes the design of innovative communication topologies \cite{orki2019control,yan2022pareto} and the analysis of the influence of the communication topologies on the platoon properties \cite{zheng2015stability,hao2011stability,hao2012achieving}. In these existing works, the topology study is separated from the controller synthesis. We believe that a better practice is to design the controller and the communication topology together, 
especially when vehicles join/leave the platoon, where the change of the topology will have an impact on the designed controller for the platoon.
Therefore, in this paper, we design the communication topology and the distributed controller simultaneously for the platoon, which is called a co-design problem.  


Our solution to the co-design problem roots in our recent work on a dissipativity-based design framework for general networked systems. Based on this framework, the platooning controller and topology co-design problem is formulated as a centralized Linear Matrix Inequality (LMI), and a decentralized way to solve such an LMI is proposed.  Unlike the existing methods, our proposed co-design method is resilient (compositional) concerning the joining/leaving of the vehicles in a platoon, i.e., there is no need to re-design the controllers for all the remaining vehicles but rather modify the controllers for those affected vehicles. It is shown that the designed controller ensures the $L_2$ stability of the closed-loop platooning system, which further implies the $L_2$ weak string stability. 
Hence, the proposed co-design approach ensures compositionality and string stability. 

This current work is closely related to our precious work on general networked dynamic systems and on a distributed backstepping controller for platoons \cite{Zihao2022Ax}. However, the novelty of this work lies in the following aspects:

\begin{enumerate}
    \item By assuming all the followers can get access to the leader's information, we propose a distributed platooning controller that can be designed in an incremental (compositional) manner based on the corresponding LMIs. 
    \item For the closed-loop platoon system under this distributed controller, we propose an optimization-based topology synthesis approach to determine the control parameters for the controller, which also represent the weighted communication links between vehicles;
    \item To support the feasibility of this controller-topology synthesis (both in the centralized and decentralized settings), we introduce local controllers at the vehicles along with a specifically developed local controller design scheme;
    \item The resulting topology (controller with certain parameters) guarantees the internal and the $L_2$ weak string stability of the closed-loop platooning system and is also security-aware and robust with respect to disturbances. 
    
\end{enumerate}

This paper is organized as follows. Section \ref{Sec:Preliminaries} summarizes several important notations and preliminary concepts about the dissipativity, networked systems, decentralization approach and string stability, where the proofs are included in Section \ref{Sec:Appendix}. The platooning problem formulation is illustrated in Section \ref{Sec:Platooning_Problem_Formulation}, followed by a centralized and decentralized co-design of the platoon in Section \ref{Sec:CentralizedCoDesign} and \ref{Sec:DecentralizedCoDesign}, respectively. In Section \ref{Sec:NumericalResults}, the simulation results for merging/splitting of the vehicles are presented. Finally, a conclusion remark is made in Section \ref{Sec:Conclusion}.


\section{Preliminaries}\label{Sec:Preliminaries}


\subsection{Notations}

The sets of real and natural numbers are denoted by $\R$ and $\N$, respectively. We define $\N_N\triangleq\{1,2,\ldots,N\}$ where $N\in\N$. 
An $n\times m$ block matrix $A$ can be represented as $A=[A_{ij}]_{i\in\N_n, j\in\N_m}$ where $A_{ij}$ is the $(i,j)$\tsup{th} block of $A$ (for indexing purposes, either subscripts or superscripts may be used, i.e., $A_{ij} \equiv A^{ij}$). 
$[A_{ij}]_{j\in \N_m}$ and $\diag([A_{ii}]_{i\in\N_n})$ represent a block row matrix and a block diagonal matrix, respectively. We define $\{A^i\} \triangleq \{A_{ii}\}\cup\{A_{ij},j\in\N_{i-1}\}\cup\{A_{ji}:j\in\N_i\}$. If $\Psi\triangleq[\Psi^{kl}]_{k,l \in \N_m}$ where each $\Psi^{kl}\triangleq[\Psi^{kl}_{ij}]_{i,j\in\N_n}$, the \emph{block element-wise form} of $\Psi$ is denoted as $\mbox{BEW}(\Psi) \triangleq [[\Psi^{kl}_{ij}]_{k,l\in\N_m}]_{i,j\in\N_n}$ (e.g., see \cite{WelikalaP32022}). 
The transpose of a matrix $A$ is denoted by $A^\T$ and $(A^\T)^{-1} = A^{-\T}$. 
The zero and identity matrices are denoted by $\0$ and $\I$, respectively (dimensions will be clear from the context). A symmetric positive definite (semi-definite) matrix $A\in\R^{n\times n}$ is represented as $A=A^\T>0$ ($A=A^\T \geq 0$). We assume $A>0 \iff A=A^\T>0$. The symbol $\star$ represents redundant conjugate matrices (e.g., $A^\T B\, \star \equiv A^\T B A$). 
The symmetric part of a matrix $A$ is defined as $\H(A) \triangleq A+A^\T$. 
$\mb{1}_{\{\cdot\}}$ is the indicator function and $\mathsf{e}_{ij} \triangleq \mb{1}_{\{i=j\}}$. 
We use $\mathcal{K}$, $\mathcal{K}_{\infty}$ and $\mathcal{KL}$ to denote different classes of comparison functions (e.g., see \cite{sontag1995characterizations}). 
For a vector $x\in\mathbb{R}^n$, the Euclidean norm of it is given by $|x|_2\triangleq |x| \triangleq \sqrt{x^\T x}$. For a time-dependent vector $x(t)\in\mathbb{R}^n$, the $\mathcal{L}_2$ and $\mathcal{L}_{\infty}$ norms of it are given by $\|x(\cdot)\|\triangleq\sqrt{\int_{0}^{\infty}|x(\tau)|^2d\tau}$ and $\|x(\cdot)\|_{\infty} \triangleq \sup_{t\geq 0}\ |x(t)|$, respectively. 

\subsection{Dissipativity} 

Consider the dynamic system  
\begin{equation}\label{Eq:GeneralSystem}
\begin{aligned}
    \dot{x}(t) = f(x(t),u(t)),\\
    y(t) = h(x(t),u(t)),
    \end{aligned}
\end{equation}
where $x(t)\in\R^{n}$, $u(t)\in \R^{q}$, $y(t)\in\R^{m}$, and $f:\R^{n}\times \R^{q} \rightarrow \R^{n}$ and $h:\R^{n}\times \R^{q}\rightarrow \R^{m}$ are continuously differentiable. 

The equilibrium points of \eqref{Eq:GeneralSystem} are such that there exists a set $\X\subset \R^{n}$ where for every $x^*\in\X$, there is a unique $u^*\in\R^{q}$ that satisfies $f(x^*,u^*)=\0$ while both $u^*$ and $y^*\triangleq h(x^*,u^*)$ being implicit functions of $x^*$. For the dissipativity analysis of \eqref{Eq:GeneralSystem} without the explicit knowledge of its equilibrium points, the \emph{equilibrium-independent dissipativity} (EID) property \cite{Arcak2022} defined next can be used. Note that it also includes the conventional dissipativity property \cite{Willems1972a}, particularly when $\X=\{\0\}$ and for $x^*=\0(\in\X)$, the corresponding $u^*=\0$ and $y^*=\0$. 

\begin{definition}\label{Def:EID}
The system \eqref{Eq:GeneralSystem} is EID under supply rate $s:\R^{q}\times\R^{m}\rightarrow \R$ if there exists a continuously differentiable storage function $V:\R^{n}\times \X \rightarrow \R$ satisfying:
$V(x,x^*)>0$ with $x \neq x^*$, 
$V(x^*,x^*)=0$, and 
$$\dot{V}(x,x^*) = \nabla_x V(x,x^*)f(x,u)\leq  s(u-u^*,y-y^*),$$ 
for all $(x,x^*,u)\in\R^{n}\times \X \times \R^{q}$.
\end{definition}

This EID property can be specialized based on the used supply rate $s(\cdot,\cdot)$. In the sequel, we use the concept of $X$-EID property \cite{WelikalaP52022}, which is defined using a quadratic supply rate determined by the coefficients matrix $X=X^\T \in\R^{q+m}$.

\begin{definition}\label{Def:X-EID}
The system \eqref{Eq:GeneralSystem} is $X$-EID (with $X \triangleq [X^{kl}]_{k,l\in\N_2}$) if it is EID under the quadratic supply rate:
\begin{center}
    $s(u-u^*,y-y^*) \triangleq  
    \bm{u-u^*\\y-y^*}^\T 
    \bm{X^{11} & X^{12}\\X^{21} & X^{22}}
    \bm{u-u^*\\y-y^*}.$
\end{center}
\end{definition}

Depending on the choice of the coefficients matrix $X$ in the above-defined $X$-EID property, it can represent several properties of interest, as detailed in the following remark.

\begin{remark}\label{Rm:X-DissipativityVersions}
If the system \eqref{Eq:GeneralSystem} is $X$-EID with: 
\begin{enumerate}
    \item $X = \scriptsize \bm{\0 & \frac{1}{2}\I \\ \frac{1}{2}\I & \0}$, then it is \emph{passive};
    \item $X = \scriptsize \bm{-\nu\I & \frac{1}{2}\I \\ \frac{1}{2}\I & -\rho\I}$, then it is \emph{strictly passive} ($\nu$ and $\rho$ are the input feedforward and output feedback passivity indices \cite{WelikalaP42022}, respectively);
    \item $X = \scriptsize \bm{\gamma^2\I & \0 \\ \0 & -\I}$, then it is \emph{$L_2$-stable}  ($\gamma$ is the $L_2$-gain, denoted as L2G($\gamma$)).
\end{enumerate}
In particular, if \eqref{Eq:GeneralSystem} is strictly passive with passivity indices $\nu$ and $\rho$ (as in the above Case 2), we also denote this as \eqref{Eq:GeneralSystem} being IF-OFP($\nu$,$\rho$). Moreover, we use the notations: IFP($\nu$)$\triangleq$IF-OFP($\nu$,0), OFP($\rho$)$\triangleq$IF-OFP(0,$\rho$) and VSP($\tau$)$\triangleq$IF-OFP($\tau,\tau$). 
\end{remark}

If the system \eqref{Eq:GeneralSystem} is linear time-invariant (LTI), a necessary and sufficient condition for it to be $X$-EID is given in the next proposition - as a linear matrix inequality (LMI) problem.

\begin{proposition}\label{Pr:LTISystemXDisspativity}
The linear time-invariant (LTI) system
\begin{equation}\label{Eq:Pr:LTISystemXDisspativity1}
\begin{aligned}
    \dot{x}(t) =&\ A x(t) + B u(t),\\
    y(t) =&\ Cx(t) + Du(t),
\end{aligned}
\end{equation}
is $X$-EID if and only if there exists $P>0$ such that 
\begin{align}\label{Eq:Pr:LTISystemXDisspativity2}
    \scriptsize \bm{
-\mathcal{H}(PA)+ C^\T X^{22} C & -PB + C^\T X^{21} + C^\T X^{22} D\\
\star & X^{11} + \mathcal{H}(X^{12} D) + D^\T X^{22}D
}\geq 0.
\end{align}
\end{proposition}

The following corollary considers a particular LTI system with a local controller (a configuration which we will encounter later on) and provides an LMI problem for synthesizing this local controller so as to enforce/optimize $X$-EID. 

\begin{corollary}\label{Co:LTISystemXDisspativation}
The LTI system
\begin{equation}\label{Eq:Co:LTISystemXDisspativation1}
    \dot{x}(t) = (A+BL)x(t) + \eta(t),
\end{equation}
is $X$-EID with $X^{22}<0$ from external input $\eta(t)$ to state $x(t)$ if and only if there exists $P>0$ and $K$ such that 
\begin{equation}\label{Eq:Co:LTISystemXDisspativation2}
\bm{-(X^{22})^{-1} & P & \0 \\
\star &-\mathcal{H}(AP+BK)& -I + PX^{21}\\
\star &\star & X^{11}}\geq 0
\end{equation}  
and $L=KP^{-1}$.
\end{corollary}

\subsection{Networked Systems}

\subsubsection{\textbf{Configuration}}
Consider the networked system $\Sigma$ shown in Fig. \ref{Fig:Interconnection2} comprised of $N$ independent dynamic subsystems $\{\Sigma_i:i\in\N_N\}$ and a static interconnection matrix $M$ that defines how the subsystems, an exogenous input signal $w(t) \in \R^{r}$ (e.g., disturbance) and an interested output signal $z(t) \in\R^{l}$ (e.g., performance) are interconnected with each other.

The dynamics of the subsystem $\Sigma_i, i\in\N_N$ are given by
\begin{equation}\label{Eq:SubsystemDynamics}
    \begin{aligned}
        \dot{x}_i(t) = f_i(x_i(t),u_i(t)),\\
        y_i(t) = h_i(x_i(t),u_i(t)),
    \end{aligned}
\end{equation}
where $x_i(t)\in\R^{n_i}, u_i(t)\in\R^{q_i}, y_i(t)\in\R^{m_i}$ and $u \triangleq \bm{u_i^\T}_{i\in\N_N}^\T, y \triangleq \bm{y_i^\T}_{i\in\N_N}^\T$. Analogous to \eqref{Eq:GeneralSystem}, each subsystem $\Sigma_i, i\in\N_N$ is assumed to have a set $\mathcal{X}_i \subset \R^{n_i}$ where for every $x_i^* \in \mathcal{X}_i$, there is a unique $u_i^* \in \R^{q_i}$ that satisfies $f_i(x_i^*,u_i^*)=0$ while both $u_i^*$ and $y_i^*\triangleq h_i(x_i^*,u_i^*)$ being implicit functions of $x_i^*$. Moreover, each subsystem $\Sigma_i, i\in\N_N$ is also assumed to be $X_i$-EID where $X_i = X_i^\T \triangleq [X_i^{kl}]_{k,l\in\N_2}$ (see Def. \ref{Def:X-EID}).

On the other hand, the interconnection matrix $M$ and the corresponding interconnection relationship are given by: 
\begin{equation}\label{Eq:NSC2Interconnection}
    \bm{u\\z}= \bm{M_{uy} & M_{uw} \\ M_{zy} & M_{zw}}\bm{y\\w} \equiv  M \bm{y\\w}.
\end{equation}
Note that, we assume $w\triangleq [w_i^\T]_{i\in\N_N}^\T$ and $z\triangleq[z_i^\T]_{i\in\N_N}^\T$, where each $w_i\in\R^{r_i}$ and $z_i\in\R^{l_i}$ respectively represent an exogenous disturbance input and an interested performance output corresponding to the subsystem $\Sigma_i,i\in\N_N$. Consequently, $w$ and $z$ are structured similarly to $u$ and $y$ in \eqref{Eq:NSC2Interconnection}.

Finally, the networked system $\Sigma$ is assumed to have a set of equilibrium states $\X \subseteq \bar{\X} \subset \R^n$ where $\bar{\X} \triangleq \X_1\times \X_2 \times \cdots \times \X_N$ and $n\triangleq\sum_{i\in\N_N} n_i$. Note that, based on the subsystem equilibrium points, for each $x^* \triangleq [(x_i^*)^\T]_{i\in\N_N}^\T \in \X$, there exists corresponding $u^*\triangleq [(u_i^*)^\T]_{i\in\N_N}^\T$ and $y^* \triangleq [(y_i^*)^\T]_{i\in\N_N}^\T$ that are implicit functions of $x^*$. Note also that, based on the networked system's equilibrium points, for each $x^* \in \X$, there is a unique $w^* \in \R^{r}$ that satisfies $u^*=M_{uy}y^*+M_{uw}w^*$ while both $w^*$ and $z^* \triangleq M_{zy}y^* + M_{zw}w^*$ being implicit functions of $x^*$. 

\vspace{-3mm}
\begin{figure}[!h]
\centering
\captionsetup{justification=centering}
\includegraphics[width=1.5in]{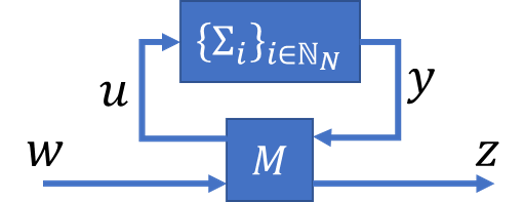}
\caption{A generic networked system $\Sigma$.}
\label{Fig:Interconnection2}
\end{figure}

\subsubsection{\textbf{Dissipativity Analysis}}
Inspired by \cite{Arcak2022}, the following proposition exploits the $X_i$-EID properties of the individual subsystems $\Sigma_i,i\in\N_N$ to formulate an LMI problem so as to analyze the $\textbf{Y}$-EID property of the networked system $\Sigma$, where $\textbf{Y}=\textbf{Y}^\T=\bm{\textbf{Y}^{kl}}_{k,l\in\N_2}$ is prespecified (e.g., based on Rm. \ref{Rm:X-DissipativityVersions}).

\begin{proposition}\label{Pr:NSC2Dissipativity}
The networked system $\Sigma$ is $\textbf{Y}$-EID if there exist scalars $p_i \geq 0, \forall i \in \N_N$ such that 
\begin{equation}\label{Eq:Pr:NSC2Dissipativity1}
\scriptsize \bm{M_{uy} & M_{uw} \\ \I & \0 \\ 
\0 & \I \\ M_{zy} & M_{zw}}^\T
\bm{\textbf{X}_p & \0 \\ \0 & -\textbf{Y}}
\bm{M_{uy} & M_{uw} \\ \I & \0 \\ 
\0 & \I \\ M_{zy} & M_{zw}} \normalsize \leq 0,
\end{equation}
where $\textbf{X}_p = [\textbf{X}_p^{kl}]_{k,l\in\N_2}$ with $\textbf{X}_p^{kl} \triangleq \diag(p_i X_i^{kl}:i\in\N_N)$.
\end{proposition}

\subsubsection{\textbf{Topology Synthesis}}
Inspired by \cite{WelikalaP52022}, the next proposition formulates an LMI problem for synthesizing the interconnection matrix $M$ \eqref{Eq:NSC2Interconnection} so as to enforce $\textbf{Y}$-EID property for the networked system $\Sigma$. However, as in \cite{WelikalaP52022}, we first need to make two mild assumptions.

\begin{figure*}[!b]
\vspace{-5mm}
\centering
\hrulefill
\begin{equation}\label{Eq:Pr:NSC2Synthesis2}
	\scriptsize \bm{
		\textbf{X}_p^{11} & \0 & L_{uy} & L_{uw} \\
		\0 & -\textbf{Y}^{22} & -\textbf{Y}^{22}M_{zy} & -\textbf{Y}^{22} M_{zw}\\ 
		L_{uy}^\T & -M_{zy}^\T\textbf{Y}^{22} & -L_{uy}^\T \textbf{X}^{12}-\textbf{X}^{21}L_{uy}-\textbf{X}_p^{22} & -\textbf{X}^{21}L_{uw}+M_{zy}^\T\textbf{Y}^{21} \\
		L_{uw}^\T & -M_{zw}^\T\textbf{Y}^{22} & -L_{uw}^\T \textbf{X}^{12}+\textbf{Y}^{12} M_{zy} &  M_{zw}^\T\textbf{Y}^{21} +  \textbf{Y}^{12}M_{zw} + \textbf{Y}^{11}
	} \normalsize >0
\end{equation}
\begin{equation}\label{Eq:Pr:NSC2Synthesis2Alternative} 
\scriptsize \bm{
 -\textbf{Y}^{22} & -\textbf{Y}^{22}M_{zy} & -\textbf{Y}^{22}M_{zw} \\
 -M_{zy}^\T\textbf{Y}^{22} & -L_{uy}^\T \textbf{X}^{12} - \textbf{X}^{21}L_{uy} - \textbf{X}_p^{22} + \alpha^2\textbf{X}_p^{11}-\alpha(L_{uy}^\T + L_{uy}) & -\textbf{X}^{21}L_{uw} + M_{zy}^\T \textbf{Y}^{21} + \alpha^2\textbf{X}_p^{11}-\alpha(L_{uy}^\T + L_{uw}) \\
-M_{zw}^\T\textbf{Y}^{22} & -L_{uw}^\T \textbf{X}^{12} + \textbf{Y}^{12}M_{zy} + \alpha^2\textbf{X}_p^{11}-\alpha(L_{uw}^\T + L_{uy}) & M_{zw}^\T \textbf{Y}^{21} +  \textbf{Y}^{12}M_{zw} + \textbf{Y}^{11} + \alpha^2\textbf{X}_p^{11}-\alpha(L_{uw}^\T + L_{uw}) 
} \normalsize >0, \alpha \in \R
\end{equation}
\end{figure*}

\begin{assumption}\label{As:NegativeDissipativity}
The given $\textbf{Y}$-EID specification for the networked system $\Sigma$ is such that $\textbf{Y}^{22}<0$. 
\end{assumption}

\begin{remark}\label{Rm:As:NegativeDissipativity}
Based on Rm. \ref{Rm:X-DissipativityVersions}, As. \ref{As:NegativeDissipativity} holds if the networked system $\Sigma$ is required to be: (1) L2G($\gamma$), or (2) OFP($\rho$) (or IF-OFP($\nu,\rho$)) with some $\rho>0$, i.e., $L_2$-stable or passive, respectively. Thus, as it is always desirable to make the networked system $\Sigma$ either $L_2$-stable or passive, As.  \ref{As:NegativeDissipativity} is mild.
\end{remark}

\begin{assumption}\label{As:PositiveDissipativity}
In the networked system $\Sigma$, each subsystem $\Sigma_i$ is $X_i$-EID with either $X_i^{11} > 0$ or $X_i^{11} < 0, \forall i\in\N_N$. 
\end{assumption}

\begin{remark}\label{Rm:As:PositiveDissipativity}
According to Rm. \ref{Rm:X-DissipativityVersions}, the above assumption directly holds if each subsystem $\Sigma_i,i\in\N_N$ is: (1) L2G($\gamma_i$) (as $X_i^{11} = \gamma_i^2 \I > 0$), or (2) IFP($\nu_i$) (or IF-OFP($\nu_i,\rho_i$)) with $\nu_i<0$ (i.e., non-passive), or (3) IFP($\nu_i$) (or IF-OFP($\nu_i,\rho_i$)) with $\nu_i>0$ (i.e., passive). Therefore, As.  \ref{As:PositiveDissipativity} is also mild. 
\end{remark}

\begin{proposition}\label{Pr:NSC2Synthesis}
Under Assumptions \ref{As:NegativeDissipativity} and \ref{As:PositiveDissipativity}, the networked system $\Sigma$ can be made $\textbf{Y}$-EID by synthesizing its interconnection matrix $M$  \eqref{Eq:NSC2Interconnection} using the LMI problem:
\begin{equation}\label{Eq:Pr:NSC2Synthesis}
    \begin{aligned}
    \mbox{Find: }& L_{uy}, L_{uw}, M_{zy}, M_{zw}, \{p_i: i\in\N_N\}\\
    \mbox{Sub. to: }& p_i > 0, \forall i\in\N_N, \mbox{ and }\\
    & \begin{cases} 
    \eqref{Eq:Pr:NSC2Synthesis2} \ \ \ \mbox{if }      &X_i^{11}>0, \forall i\in\N_N,\\ 
    \eqref{Eq:Pr:NSC2Synthesis2Alternative} \ \ \ \mbox{else if}  &X_i^{11}<0, \forall i\in\N_N, 
    \end{cases}
    \end{aligned}
\end{equation}
where $\textbf{X}^{12} \triangleq \diag((X_i^{11})^{-1}X_i^{12}:i\in\N_N)$, $\textbf{X}^{21} \triangleq (\textbf{X}^{12})^\T$ with 
$M_{uy} \triangleq (\textbf{X}_p^{11})^{-1} L_{uy}$ and $M_{uw} \triangleq  (\textbf{X}_p^{11})^{-1} L_{uw}$.
\end{proposition}

Finally, we point out that the proposed synthesis technique for the interconnection matrix $M$ given in Prop. \ref{Pr:NSC2Synthesis} can be used even when $M$ is partially known/fixed - as it will only reduce the number of variables in the LMI problem \eqref{Eq:Pr:NSC2Synthesis}. Note also that, both the proposed analysis technique given in Prop. \ref{Pr:NSC2Dissipativity} and the synthesis technique given in Prop. \ref{Pr:NSC2Synthesis} are independent of the equilibrium points of the considered networked system.

\subsection{Decentralization}

To decentrally evaluate different LMI problems arising related to networked systems (e.g., for topology synthesis:  \eqref{Eq:Pr:NSC2Synthesis}), in this section, we recall the concept of ``network matrices'' first introduced in \cite{WelikalaP32022} along with some of its useful properties.

We denote a generic directed network topology as $\mathcal{G}_n=(\mathcal{V},\mathcal{E})$ where $\mathcal{V} \triangleq \{\Sigma_i:i\in\N_n\}$ is the set of nodes (subsystems), $\mathcal{E} \subset \mathcal{V}\times \mathcal{V}$ is the set of edges (interconnections), and $n\in\N$. Corresponding to a network topology $\mathcal{G}_n$, a class of matrices named ``network matrices'' \cite{WelikalaP32022} is defined as follows.

\begin{definition}\cite{WelikalaP32022}\label{Def:NetworkMatrices}
	Corresponding to a network topology $\mathcal{G}_n=(\mathcal{V},\mathcal{E})$, any $n\times n$ block matrix $\Theta = \bm{\Theta_{ij}}_{i,j\in\N_n}$ is a \emph{network matrix} if: 
	(1) $\Theta_{ij}$ consists of information specific only to the subsystems $\Sigma_i$ and $\Sigma_j$, and 
	(2) $(\Sigma_i,\Sigma_j) \not\in \mathcal{E}$ and $(\Sigma_j,\Sigma_i)\not\in \mathcal{E}$ implies $\Theta_{ij}=\Theta_{ji}=\0$, for all $i,j\in\N_n$.
\end{definition}

To provide some examples, consider a network topology $\mathcal{G}_n=(\mathcal{V},\mathcal{E})$. First, note that, any $n \times n$ block matrix $\Theta=[\Theta_{ij}]_{i,j\in\N_n}$ where each $\Theta_{ij}, i,j\in\N_n$ is a coupling weight matrix specific only to the edge $(\Sigma_i,\Sigma_j)\in\mathcal{E}$, is a network matrix. Consequently, any corresponding adjacency-like matrix is also a network matrix. Moreover, any $n \times n$ block diagonal matrix $\Theta=\diag([\Theta_{ii}]_{i\in \N_n})$ where each $\Theta_{ii}, i\in\N_n$ is specific only to the subsystem $\Sigma_i\in\mathcal{V}$, is a network matrix. 

The following proposition, compiled from \cite{WelikalaP32022}, summarizes several useful properties of such network matrices.

\begin{proposition}
\label{Pr:NetworkMatrixProperties}
\cite{WelikalaP32022}
Given a network topology $\mathcal{G}_n$, if $\Theta$ and $\Phi$ are two corresponding network matrices, then: 
\newline
\textbf{1)}\ $\Theta^\T$ and $\alpha \Theta + \beta \Phi$ are network matrices for any $\alpha,\beta \in \R$;\newline 
\textbf{2)}\ $\Phi \Theta$ and $\Theta\Phi$ are network matrices whenever $\Phi$ is a block diagonal network matrix.\newline
Moreover, if $\Psi = [\Psi^{kl}]_{k,l\in\N_m}$ is a block matrix where each block $\Psi^{kl},k,l\in\N_m$ is a network matrix of $\mathcal{G}_n$, then: 
\newline
\textbf{3)}\ $\mbox{BEW}(\Psi)\triangleq [[\Psi^{kl}_{ij}]_{k,l\in\N_m}]_{i,j\in\N_n}$ is a network matrix;\newline
\textbf{4)}\ $\Psi >0 \iff \mbox{BEW}(\Psi)>0$.
\end{proposition}

The above proposition can be used to identify conditions under which certain matrices become network matrices. For example, if $A$ and $P$ are block network matrices and $P$ is also block diagonal, then: $A^\T P$, $P A$, $A^\T P + PA$ and $\mbox{BEW}(\scriptsize \bm{P & A^\T P\\ PA & P})$ are all network matrices. Moreover, it can be used to transform certain block matrix inequalities of interest into network matrix inequalities, for example, 
$$\scriptsize \bm{P & A^\T P\\ PA & P}>0 \iff \normalsize \mbox{BEW}(\scriptsize \bm{P & A^\T P\\ PA & P}) >0.$$

The following proposition (inspired by Sylvester's criterion \cite{Antsaklis2006}, taken from \cite{WelikalaP32022}) provides an approach to sequentially analyze/enforce a positive definiteness condition imposed on a symmetric block matrix.

\begin{proposition}
\label{Pr:MainProposition}
\cite{WelikalaP32022}
A symmetric $N \times N$ block matrix $W = [W_{ij}]_{i,j\in\N_N} > 0$ if and only if
\begin{equation}\label{Eq:Pr:MainProposition1}
    \tilde{W}_{ii} \triangleq W_{ii} - \tilde{W}_i \mathcal{D}_i \tilde{W}_i^\T > 0, \ \ \ \ \forall i\in\N_N,
\end{equation}
where $\tilde{W}_i \triangleq [\tilde{W}_{ij}]_{j\in\N_{i-1}} \triangleq W_i(\mathcal{D}_i\mathcal{A}_i^\T)^{-1}$, 
$W_i \triangleq  [W_{ij}]_{j\in\N_{i-1}}$,
$\mathcal{D}_i \triangleq \diag(\tilde{W}_{jj}^{-1}:j\in\N_{i-1})$, and $\mathcal{A}_i$ is the block lower-triangular matrix created from $[\tilde{W}_{kl}]_{k,l\in\N_{i-1}}$.
\end{proposition}

\begin{remark}\label{Rm:Pr:MainProposition}
If $W$ is a symmetric block network matrix corresponding to some network topology $\mathcal{G}_N$, the above proposition can be used to analyze/enforce $W>0$ in a decentralized (as well as compositional) manner by sequentially analyzing/enforcing $\tilde{W}_{ii}>0$ at each subsystem $\Sigma_i, i\in\N_N$. Moreover, this decentralized matrix inequality $\tilde{W}_{ii}>0$ can be transformed into an LMI in $[W_{ij}]_{j\in\N_{i}}$ using the Schur complement theory. Therefore, such a decentralized analysis/enforcement can be executed efficiently using readily available convex optimization toolboxes \cite{Boyd1994}. Additional details on Prop. \ref{Pr:MainProposition} are omitted here but can be found in \cite{WelikalaP32022,Welikala2022Ax2}.
\end{remark}


\subsection{String Stability Analysis}

To conclude this section, here we recall two ``String Stability'' concepts that can be used to characterize how disturbances propagate over the networked systems. 

Similar to before (e.g., see \eqref{Eq:SubsystemDynamics}), consider a networked system $\Sigma$ comprised of $N$ subsystems $\{\Sigma_i: i\in\mathbb{N}_N\}$, where the dynamics of each subsystem $\Sigma_i, i\in\N_N$ are given by 
\begin{equation}\label{Eq:subsystem_i_networked_system}
    \begin{split}
        \dot{x}_i=f_i(x_i,\{x_j\}_{j\in\mathcal{E}_i},w_i),
    \end{split}
\end{equation}
where $x_i\in\mathbb{R}^{n_i}$ and $w_i\in\mathbb{R}^{r_i}$ are respectively the subsystem's state and disturbance. $\{x_j: j\in\mathcal{E}_i\}$ are the states of the neighboring subsystems of the subsystem $\Sigma_i$. Consequently, the dynamics of the networked system $\Sigma$ can be written as 
\begin{equation}\label{Eq:networked_system_platoon}
\dot{x} = f(x,w)    
\end{equation}
where $x\triangleq [x_i^\T]_{i\in\N_N}^\T$, $w\triangleq [w_i^\T]_{i\in\N_N}^\T$, $f\triangleq[f_i^\T]_{i\in\N_N}^\T:\mathbb{R}^{n}\times\mathbb{R}^{r}\rightarrow\mathbb{R}^{n}$, $n \triangleq \sum_{i\in\mathbb{N}_N} n_i$ and $r \triangleq \sum_{i\in\mathbb{N}_N} r_i$. The function $f$ is assumed to be such that $f(x^*,\0) = \0, \forall x^* \in \X \subset \R^n$ ($\X$ denotes a set of equilibrium states), and locally Lipschitz continuous around each equilibrium point $x^* \in \X$. 

To characterize the string stability properties of the networked system \eqref{Eq:networked_system_platoon}, we first recall the definition of \emph{$L_2$ weak string stability} initially proposed in \cite{ploeg2013lp}. 
\begin{definition}\cite{ploeg2013lp}\label{def:l2_weak_ss}
    The networked system \eqref{Eq:networked_system_platoon} around the equilibrium point $x^* \in \X$ is $L_2$ weakly string stable if there exist class $\mathcal{K}$ functions $\alpha_1$ and $\alpha_2$ such that, for any initial state $x(0)\in\mathbb{R}^n$ and $w(\cdot)\in\mathcal{L}_2^r$, the condition: 
\begin{equation}\label{Eq:l2_wss_condition}
    \|x(\cdot)-x^*\|\leq\alpha_1(\|w(\cdot)\|)+\alpha_2(|x(0)-x^*|)
\end{equation}
holds for any $N\in\mathbb{N}$.
\end{definition}

\begin{remark}
    The above defined $L_2$ weak string stability is called ``weak'' as it allows the norm 
    $\|x(\cdot)-x^*\|$ to increase locally with respect to $N$ (i.e., the total number of subsystems in the network) while guaranteeing the boundedness of its maximum.    
    In particular, such a maximum is bounded independently of $N$, as in such a case, the condition \eqref{Eq:l2_wss_condition} needs to hold for any $N\in\N$. 
    However, $L_2$ weak string stability does not require any bounds on the external disturbances and the initial state perturbations. 
\end{remark}

\section{Problem Formulation and Basic Setup}\label{Sec:Platooning_Problem_Formulation}
\subsection{Platooning Problem Formulation} 

Following \cite{Zihao2022b}, we consider the longitudinal dynamics of the $i$\tsup{th} vehicle $\Sigma_i, i \in \bar{\N}_{N} \triangleq \N_N \cup \{0\}$ in the platoon to be 
\begin{equation}\label{Eq:VehicleDynamics}
\Sigma_i :
\begin{cases}
    \dot{x}_i(t) = v_i(t) + d_{xi}(t),\\ 
    \dot{v}_i(t) = a_i(t) + d_{vi}(t),\\
    \dot{a}_i(t) = f_i(v_i(t),a_i(t)) + \frac{1}{m_i\tau_i}\bar{u}_i(t) + d_{ai}(t),
\end{cases}
\end{equation}
where $f_i \triangleq f_i(v_i(t),a_i(t))$ with 
\begin{equation*}
f_i \triangleq -\frac{1}{\tau_i} \left(a_i(t)+C_{ri} + \frac{\rho A_{fi} C_{di} v_i(t)}{2m_i}(v_i(t)+2\tau_ia_i(t)) \right).
\end{equation*}
Note that here we use $\Sigma_0$ to represent the leading vehicle (virtual) and $\{\Sigma_i: i\in\N_N\}$ to represent the following $N$ vehicles (controllable) in the platoon. In \eqref{Eq:VehicleDynamics}, the position, velocity and acceleration of the vehicle $\Sigma_i, i\in\bar{N}_N$ are denoted by $x_i(t), v_i(t)$, $a_i(t)\in\mathbb{R}$, and the corresponding disturbances are denoted by $d_{xi}(t), d_{vi}(t)$, $d_{ai}(t) \in \mathbb{R}$, respectively. Further, $\bar{u}_i(t)\in\R$ in \eqref{Eq:VehicleDynamics} is the control input to be designed for the follower vehicle $\Sigma_i, i\in\N_N$, and it is prespecified for the leading vehicle $\Sigma_0$. The remaining coefficients appearing in \eqref{Eq:VehicleDynamics} are as follows: $\rho$ is the air density, $m_i,\tau_i$ and $A_{fi}$ are the mass, the engine time constant and the effective frontal area of $\Sigma_i$, respectively, and $C_{di}$ and $C_{ri}$ are the coefficients of aerodynamic drag and rolling resistance of $\Sigma_i$, respectively.

\paragraph*{\textbf{Linearizing State Feedback Control}} 
Using the fact that $m_i,\tau_i \neq 0$ in \eqref{Eq:VehicleDynamics}, a state feedback linearizing control law for $\bar{u}_i(t)$ can be designed as  
\begin{equation}\label{Eq:LinearizingControlInput}
    \bar{u}_i(t) = m_i\tau_i \left(-f_i(v_i(t),a_i(t))+g_i(u_i(t))\right),
\end{equation}
where $g_i(u_i(t))$ is the virtual control input law that needs to be designed for the follower vehicles $\Sigma_i, i\in\N_N$. For notational convenience, we assume that \eqref{Eq:LinearizingControlInput} is also implemented at the leading vehicle $\Sigma_0$ and $u_0(t)$ is prespecified.

Consequently, under \eqref{Eq:LinearizingControlInput}, the longitudinal dynamics of the vehicle $\Sigma_i, i\in \bar{\N}_N$ take the triple integrator form:
\begin{equation}\label{Eq:VehicleDynamicsLinearized}
    \Sigma_i : 
    \begin{cases}
    \dot{x}_i(t) = v_i(t) + d_{xi}(t),\\ 
    \dot{v}_i(t) = a_i(t) + d_{vi}(t),\\
    \dot{a}_i(t) = g_i(u_i(t)) + d_{ai}(t).\\
    \end{cases}
\end{equation}

Note that the individual vehicle dynamics \eqref{Eq:VehicleDynamicsLinearized} (or \eqref{Eq:VehicleDynamics}) are independent of each other. However, they become dependent on each other through the used virtual feedback control law $g_i(u_i(t))$ - which is typically designed using the state of some corresponding \emph{error dynamics}. In this paper, we consider two different error dynamics formulations as described next. 

In this paper, our goal is to propose a co-design framework for the distributed controller and communication topology of a platoon, where each vehicle in the platoon is with the dynamics \eqref{Eq:VehicleDynamicsLinearized}. Besides, the string stability and the compositionality of the platoon are both enforced to be guaranteed.

\subsection{Platooning Error Dynamics}

\subsubsection{\textbf{Error Dynamics I}}
Based on the other vehicle locations $\{x_j(t):j\in\bar{\N}_N \backslash \{i\}\}$, here we define the \emph{location error} of the vehicle $\Sigma_i, i\in\N_N$ as
\begin{equation}\label{Eq:LocationError}
    \tilde{x}_i(t) \triangleq \sum_{j\in\bar{\N}_N} \bar{k}_{ij} \left( x_i(t)-x_j(t)-d_{ij} \right),
\end{equation}
where $\bar{k}_{ij}\in\R$ is a weighting coefficient and $d_{ij} \in \R$ is the desired separation between vehicles $\Sigma_i$ and $\Sigma_j$. Note that $\bar{k}_{ii} \triangleq 0$ and $d_{ii} \triangleq 0$ for all $i\in\N_N$. Note also that, if vehicle $\Sigma_i$ gets information from vehicle $\Sigma_j$ (through some communication medium), then, $\bar{k}_{ij}>0$, otherwise, $\bar{k}_{ij}=0$. On the other hand, as shown in Fig. \ref{Fig:PlatoonPlacement}, $d_{ij} \triangleq d_i-d_j$ with $d_0 \triangleq 0$ and $d_m \triangleq \sum_{k\in\N_m} L_k+\delta_{k}$ for $m=i,j$, where $L_k$ is the length of the vehicle $\Sigma_k$ and $\delta_{k}$ is the desired separation between vehicles $\Sigma_{k}$ and $\Sigma_{k-1}$. 

\begin{figure}[!b]
    \centering
    \includegraphics[width=3.4in]{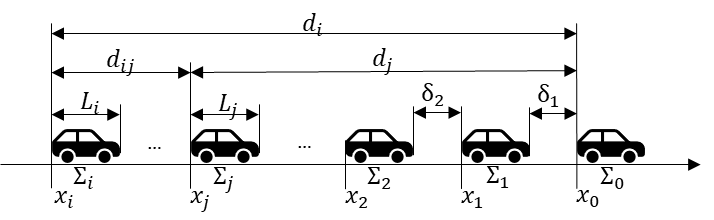}
    \caption{Vehicle placements in the platoon.}
    \label{Fig:PlatoonPlacement}
\end{figure}

Taking the derivative of  \eqref{Eq:LocationError} yields
\begin{equation}\label{Eq:LocationErrorDynamics}
\dot{\tilde{x}}_{i}(t) = \tilde{v}_{i}(t) + \bar{d}_{xi}(t),
\end{equation}
where $\bar{d}_{xi}(t)\triangleq \sum_{j\in\bar{\N}_N} \bar{k}_{ij} \left(d_{xi}(t)-d_{xj}(t)\right)$ is the disturbance associated with the location error dynamics and 
\begin{equation}\label{Eq:VelocityError}
    \tilde{v}_{i}(t)\triangleq \sum_{j\in\bar{\N}_N} \bar{k}_{ij} \left(v_i(t)-v_j(t)\right)
\end{equation} 
is the \emph{velocity error} of the vehicle $\Sigma_i, i\in\N_N$.

Next, taking the derivative of \eqref{Eq:VelocityError} yields
\begin{equation}\label{Eq:VelocityErrorDynamics0}
    \dot{\tilde{v}}_{i}(t) = \sum_{j\in\bar{\N}_N} \bar{k}_{ij} \left(a_i(t)-a_j(t)\right) + \bar{d}_{vi}(t),
\end{equation}
where $\bar{d}_{vi}(t)\triangleq \sum_{j\in\bar{\N}_N} \bar{k}_{ij} \left(d_{vi}(t)-d_{vj}(t)\right)$ is the disturbance associated with the velocity error dynamics. 
However, unlike in the previous step, here we do not define the first term in \eqref{Eq:VelocityErrorDynamics0} as the \emph{acceleration error} due to its interconnected and sensitive nature. Instead, we define the \emph{acceleration error} of the vehicle $\Sigma_i, i\in\N_N$ as
\begin{equation}\label{Eq:AccelerationError}
    \tilde{a}_{i}(t) \triangleq \left(a_i(t) - a_0(t)\right).
\end{equation}
Therefore, \eqref{Eq:VelocityErrorDynamics0} can be re-stated as
\begin{equation}
    \dot{\tilde{v}}_{i}(t) = \sum_{j\in \N_N}k_{ij}\tilde{a}_{j}(t) + \bar{d}_{vi}(t),
\end{equation}
where 
\begin{equation}\label{Eq:bTocCoefficients}
    k_{ii} \triangleq \bar{k}_{i0}+\sum_{j\in \N_N} \bar{k}_{ij} \ \mbox{ and } \ 
    k_{ij} \triangleq -\bar{k}_{ij}, \forall j\in \N_N \backslash \{i\}.
\end{equation} 
It is worth noting that the coefficients $\{k_{ij}:j\in\N_N\}$ can uniquely determine the coefficients $\{\bar{k}_{ij}:j\in\bar{\N}_N\}$ as $\bar{k}_{ii}\triangleq 0$.

Finally, taking the derivative of \eqref{Eq:AccelerationError} yields
\begin{equation}
    \dot{\tilde{a}}_{i}(t) = g_i(u_i(t)) + \bar{d}_{ai}(t)
\end{equation}
where $\bar{d}_{ai}(t)\triangleq - g_0(u_0(t)) + \left(d_{ai}(t)-d_{a0}(t)\right)$ is the disturbance associated with the acceleration error dynamics. Note that, here we have considered the leader's virtual control input law $g_0(u_0(t))$ as a disturbance, because, typically, $g_0(u_0(t))=u_0=0$ holds for all $t\geq 0$ except for a few time instants where the leader changes its acceleration to a different (still constant) level, e.g., see \cite{Zihao2022Ax}.

In all, the error dynamics of the vehicle $\Sigma_i, i\in\N_N$ are:
\begin{equation}\label{Eq:VehicleErrorDynamics}
\begin{cases}
\dot{\tilde{x}}_{i}(t) =&\ \tilde{v}_{i}(t) + \bar{d}_{xi}(t),\\
\dot{\tilde{v}}_{i}(t) =&\ \sum_{j\in \N_N}k_{ij}\tilde{a}_{j}(t) + \bar{d}_{vi}(t),\\
\dot{\tilde{a}}_{i}(t) =&\ g_i(u_i(t)) + \bar{d}_{ai}(t).
\end{cases}
\end{equation}
We select the virtual control input law $g_i(u_i(t))$ in the form
\begin{equation}\label{Eq:VirtualControlInputLawForm}
    g_i(u_i(t)) \triangleq \bar{L}_{ii}e_i(t) + u_i(t),
\end{equation}
where $\bar{L}_{ii}\in\R^{1\times 3}$ is a \textbf{local controller gain} and $e_i(t) \triangleq \bm{\tilde{x}_{i}(t) & \tilde{v}_{i}(t) & \tilde{a}_{i}(t)}^\T$. Now, by defining  
\begin{equation}\label{Eq:VehicleErrorDynamicsAMat1}
    A \triangleq \bm{0&1&0\\0&0&1\\0&0&0}, \quad B \triangleq \bm{0 \\ 0 \\ 1},
\end{equation}
we can restate \eqref{Eq:VehicleErrorDynamics} (under the choice \eqref{Eq:VirtualControlInputLawForm}) as 
\begin{equation}\label{Eq:VehicleErrorDynamicsVectored}
    \dot{e}_i(t) = (A+B\bar{L}_{ii})e_i(t) + B u_i(t)
    +\sum_{j\in\N_N} \bar{K}_{ij} e_{j}(t) + \bar{d}_i^{(1)}(t),
\end{equation}
where (recall: $\mathsf{e}_{ij} \triangleq \mb{1}_{\{i=j\}}$)
\begin{equation}     
    \bar{K}_{ij} \triangleq \bm{0&0&0\\0&0&k_{ij}-\mathsf{e}_{ij}\\0&0&0},
    \ \ \ \ 
    \bar{d}_i^{(1)}(t) \triangleq \bm{\bar{d}_{xi}(t) \\ \bar{d}_{vi}(t) \\ \bar{d}_{ai}(t)}.    
\end{equation}

\paragraph*{\textbf{Feedback Control}} 
To control the error dynamics \eqref{Eq:VehicleErrorDynamicsVectored}, we select the virtual control input $u_i(t)$ as
\begin{equation}\label{Eq:PlatoonFeedbackControl}
    u_i(t) \triangleq L_{ii} e_i(t), 
\end{equation}
where $L_{ii} \triangleq \bm{l_{ii}^x & l_{ii}^v & l_{ii}^a} \in \R^{1\times 3}$ is a \textbf{global controller gain} (besides $\{k_{ij}:j\in\N_N\}$, see also \eqref{Eq:VirtualControlInputLawForm}). Now, the closed-loop error dynamics of the vehicle $\Sigma_i,i\in\N_N$ take the form
\begin{equation}\label{Eq:PlatoonClosedLoop1}
\begin{aligned}
    \tilde{\Sigma}_{i}: \quad \begin{cases} \dot{e}_i(t) = (A + B\bar{L}_{ii})e_i(t) + \eta_i(t),\end{cases}
\end{aligned}
\end{equation}
where 
\begin{equation}\label{Eq:PlatoonClosedLoop1Eta}
    \begin{aligned}
        \eta_i(t) \triangleq&\ \sum_{j\in\N_N} K_{ij} e_{j}(t) + w_i(t),
    \end{aligned}
\end{equation}
with 
\begin{equation}\label{Eq:PlatoonClosedLoopK1}
K_{ij} \triangleq \bm{0&0&0\\0&0&k_{ij}-\mathsf{e}_{ij}\\l_{ii}^x \mathsf{e}_{ij} & l_{ii}^v \mathsf{e}_{ij} & l_{ii}^a \mathsf{e}_{ij}}, 
\ \ \ \ 
w_i(t)\triangleq \bar{d}_i^{(1)}(t).
\end{equation}

\subsubsection{\textbf{Error Dynamics II}}\label{sec:ErrorDynamicsII}
Compared to \eqref{Eq:LocationError}, here we define the \emph{location error} of the vehicle $\Sigma_i,i\in\N_N$ as
\begin{equation}\label{Eq:LocationError2}
    \tilde{x}_i(t) \triangleq  x_i(t) - (x_0(t)-d_{i0}),
\end{equation}
where $d_{i0}$ is the desired separation between vehicles $\Sigma_i$ and $\Sigma_0$ (note that we use the same notation as before).  

Taking the derivative of \eqref{Eq:LocationError2} yields
\begin{equation}\label{Eq:LocationErrorDynamics2}
    \dot{\tilde{x}}(t) = \tilde{v}_i(t) + d_{xi}(t),
\end{equation}
where 
\begin{equation}\label{Eq:VelocityError2}
    \tilde{v}_i(t) \triangleq v_i(t) - v_0(t) - d_{x0}(t),   
\end{equation}
is the \emph{velocity error} of the vehicle $\Sigma_i, i\in\N_N$. Note that $d_{xi}(t)$ and $d_{x0}(t)$ are directly from \eqref{Eq:VehicleDynamicsLinearized} (or \eqref{Eq:VehicleDynamics}). 

Next, taking the derivative of \eqref{Eq:VelocityError2} yields
\begin{equation}\label{Eq:VelocityErrorDynamics2}
    \dot{\tilde{v}}_i(t) = \tilde{a}(t) + d_{vi}(t),
\end{equation}
where 
\begin{equation}\label{Eq:AccelerationError2}
    \tilde{a}_i(t) \triangleq a_i(t) - a_0(t) - d_{v0}(t) - \dot{d}_{x0}(t),
\end{equation}
is the \emph{acceleration error} of the vehicle $\Sigma_i, i\in\N_N$.  

Finally, taking the derivative of \eqref{Eq:AccelerationError2} yields
\begin{equation}\label{Eq:AccelerationErrorDynamics2}
    \dot{\tilde{a}}_i(t) = g_i(u_i(t)) + d_{ai}(t) - \tilde{u}_0(t),
\end{equation}
where 
\begin{equation}
    \tilde{u}_0(t) \triangleq g_0(u_0(t)) + d_{a0}(t) + \dot{d}_{v0}(t) + \ddot{d}_{x0}(t),
\end{equation}
can be thought of as a disturbance induced by the leader's motion. Note that, if we constrained the leader's motion by assuming it is noiseless and $\dot{a}_0(t)=0, \forall t\in\R_{\geq 0}$, then, $\tilde{u}_0(t)=0, \forall t\in\R_{\geq 0}$.

In all, the error dynamics of the vehicle $\Sigma_i,i\in\N_N$ are:
\begin{equation}\label{Eq:VehicleErrorDynamics2} 
    \begin{cases}
    \dot{\tilde{x}}(t) = \tilde{v}_i(t) + d_{xi}(t), \\
    \dot{\tilde{v}}_i(t) = \tilde{a}_i(t) + d_{vi}(t),\\
    \dot{\tilde{a}}_i(t) = g_i(u_i(t)) + d_{ai}(t) - \tilde{u}_0(t).
    \end{cases}
\end{equation}

Now, using the same definitions in \eqref{Eq:VirtualControlInputLawForm} (that includes a \textbf{local controller gain} $\bar{L}_{ii}$) and \eqref{Eq:VehicleErrorDynamicsAMat1}, we can restate \eqref{Eq:VehicleErrorDynamics2} as 
\begin{equation}\label{Eq:VehicleErrorDynamicsVectored2}
    \dot{e}_i(t) = (A+B\bar{L}_{ii})e_i(t) + Bu_i(t) + \bar{d}_i^{(2)}(t),
\end{equation}
where $\bar{d}_i^{(2)}(t) \triangleq \bm{d_{xi}(t)&d_{vi}(t)&d_{ai}(t)-\tilde{u}_0(t)}^\T$.
Note that, unlike in \eqref{Eq:VehicleErrorDynamicsVectored}, the error dynamics obtained in \eqref{Eq:VehicleErrorDynamicsVectored2} are not coupled (with those of other follower vehicles). This is due to the decoupled location \eqref{Eq:LocationError2}, velocity \eqref{Eq:VelocityError2}, and acceleration \eqref{Eq:AccelerationError2} errors considered in the derivation of \eqref{Eq:VehicleErrorDynamicsVectored2}.

\paragraph*{\textbf{Feedback Control}} 
To control the error dynamics \eqref{Eq:VehicleErrorDynamicsVectored2}, instead \eqref{Eq:PlatoonFeedbackControl}, here we select the virtual control input $u_i(t)$ as  
\begin{equation}
    \label{Eq:PlatoonFeedbackControl2}
    u_i(t) = L_{ii} e_i(t) + \sum_{j\in\N_N\backslash\{i\}} L_{ij} (e_i(t)-e_j(t)), 
\end{equation}
where $L_{ij} \triangleq \bm{l_{ij}^x & l_{ij}^v & l_{ij}^a}\in\R^{1\times 3}, \forall j\in \N_N$ are \textbf{global controller gains}. For notational convenience, let us restate \eqref{Eq:PlatoonFeedbackControl2} as 
\begin{equation}
    u_i(t) = \sum_{j\in\N_N}\bar{K}_{ij} e_j(t)
\end{equation}
where $\bar{K}_{ij} \triangleq -L_{ij}, \forall j\neq i$ and $\bar{K}_{ii} \triangleq L_{ii} + \sum_{j\in\N_N\backslash\{i\}}L_{ij}$. 
Finally, by defining $\bm{k_{ij}^x&k_{ij}^v&k_{ij}^a} \triangleq \bar{K}_{ij}, \forall j\in\N_N$, the closed-loop error dynamics of the vehicle $\Sigma_i,i\in\N_N$ can be shown to have the same form as in \eqref{Eq:PlatoonClosedLoop1}-\eqref{Eq:PlatoonClosedLoop1Eta}, but with 
\begin{equation}\label{Eq:PlatoonClosedLoopK2}
K_{ij} \triangleq \bm{0&0&0\\0&0&0\\k_{ij}^x&k_{ij}^v&k_{ij}^a},
\quad
w_i(t) \triangleq \bar{d}_i^{(2)}(t),  
\end{equation}
as opposed to the choices saw in \eqref{Eq:PlatoonClosedLoopK1}.

\begin{remark}
For a fully connected vehicular platoon, the error dynamics formulations I and II presented above respectively have $N^2+6$ and $3N^2+3$ scalar controller gain parameters to be synthesized. This is due to their respective difference in \eqref{Eq:PlatoonClosedLoopK1} and \eqref{Eq:PlatoonClosedLoopK2} (despite sharing the same closed-loop error dynamics \eqref{Eq:PlatoonClosedLoop1}-\eqref{Eq:PlatoonClosedLoop1Eta}). As a consequence, we can expect the latter approach to be less conservative while the prior approach is more computationally inexpensive to synthesize. 
\end{remark}

\subsection{Networked System View}

Irrespective of the used error dynamics formulation, both resulting closed-loop error dynamics share a common format \eqref{Eq:PlatoonClosedLoop1}-\eqref{Eq:PlatoonClosedLoop1Eta} that can be modeled as a networked system as shown in Fig. \ref{Fig:PlatoonProblem} (analogous to Fig. \ref{Fig:Interconnection2}). 

In particular, here, each subsystem $\tilde{\Sigma}_i, i\in\N_N$ in the networked system (denoted $\tilde{\Sigma}$) is an LTI system $\tilde{\Sigma}_i:\eta_i\rightarrow e_i$ as given in \eqref{Eq:PlatoonClosedLoop1}-\eqref{Eq:PlatoonClosedLoop1Eta}. Let us collectively denote the subsystem inputs as $\eta(t)\triangleq [\eta_i^\T(t)]_{i\in\N_N}$ and subsystem outputs as $e(t) \triangleq [e_i^\T(t)]_{i\in\N_N}$. In addition, the overall networked system $\tilde{\Sigma}$ also include a disturbance input $w(t) \triangleq [w_i^\T(t)]_{i\in\N_N}$ and a performance output $z(t) = [z_i^\T(t)]_{i\in\N_N}$. These four components/ports are interconnected according to the relationship:
\begin{equation}\label{Eq:PlatoonInterconnection}
    \bm{\eta\\z} = \bm{M_{\eta e} & M_{\eta w}\\M_{ze} & M_{zw}}\bm{e\\w} \equiv M\bm{e\\w},
\end{equation}
where $M$ is the interconnection matrix, $M_{\eta e} \triangleq [K_{ij}]_{i,j\in\N_N}$ (from \eqref{Eq:PlatoonClosedLoopK1} or \eqref{Eq:PlatoonClosedLoopK2}), $M_{\eta w} \triangleq \I$, $M_{ze} \triangleq \I$, and $M_{zw} \triangleq \0$.  

Based on this networked system view of the closed-loop error dynamics of the platoon, we can now apply Prop. \ref{Pr:NSC2Synthesis} to synthesize a suitable $M_{\eta e} \triangleq [K_{ij}]_{i,j\in\N_N}$ while constraining each $K_{ij}, i,j\in\N_N$ to be of the form \eqref{Eq:PlatoonClosedLoopK1} or \eqref{Eq:PlatoonClosedLoopK2} (depending on the used error dynamics formulation). Consequently, it can be used to obtain the corresponding controller gains required to evaluate the virtual control inputs \eqref{Eq:PlatoonFeedbackControl} or \eqref{Eq:PlatoonFeedbackControl2}. Finally, we conclude this section by pointing out that synthesizing $M_{\eta e}$ will not only reveal the desired individual controllers but also a preferable communication topology for the platoon. 

\vspace{-3mm}
\begin{figure}[!h]
    \centering
    \includegraphics[width=3in]{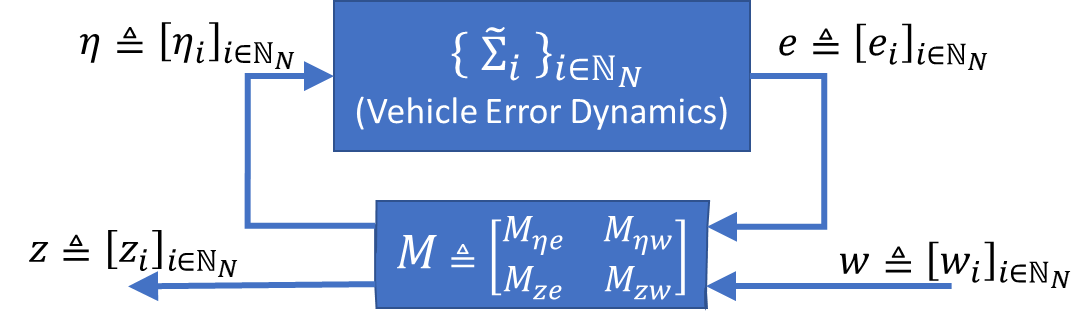}
    \caption{Platoon error dynamics as a networked system $\tilde{\Sigma}$.}
    \label{Fig:PlatoonProblem}
\end{figure}

\section{Centralized Design of the Platoon}\label{Sec:CentralizedCoDesign}

In this section, we discuss centralized LMI-based techniques to design the platoon, i.e., to obtain the local controller gains, the global controller gains and the communication topology.

As detailed earlier, since we view the overall closed-loop error dynamics of the platoon as a networked system $\tilde{\Sigma}$ (see Fig. \ref{Fig:PlatoonProblem}) and intend to apply Prop. \ref{Pr:NSC2Synthesis} for the global controller and topology design, we require the dissipativity properties of the individual subsystems, i.e., of the individual vehicle closed-loop error dynamics $\tilde{\Sigma}_i, i\in\N_N$ (see \eqref{Eq:PlatoonClosedLoop1}). However, based on \eqref{Eq:PlatoonClosedLoop1} and Co. \ref{Co:LTISystemXDisspativation}, it is clear that such subsystem dissipativity properties will be heavily dependent upon the used local controllers $\bar{L}_{ii}, i\in\N_N$ (see \eqref{Eq:VirtualControlInputLawForm}). Consequently, designing local controllers and global controllers (along with the topology) are two highly coupled problems.

However, solving such two coupled problems together as a single combined problem is not practical as such a combined problem will be a significantly larger, non-decentralizable and non-linear matrix inequality problem - as opposed to being a smaller, decentralizable and linear matrix inequality problem (e.g., like the problems in \eqref{Eq:Co:LTISystemXDisspativation2} and \eqref{Eq:Pr:NSC2Dissipativity1}).

On the other hand, we also cannot solve these two coupled problems as two entirely independent/decoupled problems. This is mainly because subsystem dissipativity properties enforced by an independent local controller design may not be sufficient/suitable for obtaing a feasible global controller design for the networked system. Conversely, subsystem dissipativity properties required for obtaining a feasible global controller design for the networked system may not be achievable by an independent local controller design. Moreover, the optimality of the overall design is also a concern if these two coupled problems are to be treated as two entirely independent/decoupled problem. 

To address this challenge, in this paper, we consider these two problems as two ``loosely coupled'' problems. In particular, we include several additional parameters and conditions in the local controller design so as to ensure that the resulting subsystem dissipativity properties are feasible as well as effective to be used in the global controller design. Moreover, the proposed loosely coupled approach to design local and global controllers is also decentralizable due to its emphasis on strengthening the local controller design at subsystems. 

This section is arranged as follows. We first present the global controller design problem for the networked system. Next, we identify several necessary conditions on the subsystem dissipativity properties so as to achieve a feasible and effective global controller design. Subsequently, we present the (strengthened) local controller design problem that enforces the identified necessary conditions. Finally, we summarize the overall local and global controller design process.

\subsection{Global Controller and Topology Synthesis}

Assume each subsystem $\tilde{\Sigma}_i,i\in\N_N$ \eqref{Eq:PlatoonClosedLoop1} to be $X_i$-EID with 
\begin{equation}\label{Eq:SubsystemDissipativity}
    X_i = \bm{X_i^{11} & X_i^{12} \\
    X_i^{21} & X_i^{22}} \triangleq \bm{-\nu_i\I & \frac{1}{2}\I \\ \frac{1}{2} \I & -\rho_i\I},
\end{equation}
i.e., IF-OFP($\nu_i$,$\rho_i$) (see Rm. \ref{Rm:X-DissipativityVersions}), where passivity indices $\nu_i$ and $\rho_i$ are fixed and known (we will relax this assumption later on). These subsystem EID properties will be used to apply Prop. \ref{Pr:NSC2Synthesis} to synthesize the interconnection matrix $M$ \eqref{Eq:PlatoonInterconnection} (particularly, its block $M_{\eta e}=[K_{ij}]_{i,j\in\N_N}$). Upon the synthesis of matrices $[K_{ij}]_{i,j\in\N_N}$, we can uniquely determine the controller gains $\{\bar{k}_{ij}:i,j\in\N_N\}$ \eqref{Eq:LocationError} or $\{L_{ij}:i,j\in\N_N\}$ \eqref{Eq:PlatoonFeedbackControl2}, depending on the used error dynamics formulation. Moreover, it will reveal the desired communication topology for the platoon. 

As the main specification for this controller and topology synthesis, we require the overall platoon error dynamics (i.e., the networked system $\tilde{\Sigma}$ in Fig. \ref{Fig:PlatoonProblem}) to be finite-gain $L_2$-stable with some $L_2$-gain $\gamma$ from its disturbance input $w(t)$ to performance output $z(t)$. Particularly, letting $\tilde{\gamma} \triangleq \gamma^2$ we require $0<\tilde{\gamma}<\bar{\gamma}$ where $\bar{\gamma}\in\R_{>0}$ is a pre-specified constant. Based on Rm. \ref{Rm:X-DissipativityVersions}, this specification is akin to requiring the networked system $\tilde{\Sigma}$ to be $\textbf{Y}$-dissipative with 
\begin{equation}\label{Eq:NetworkedSystemDissipativity}
\textbf{Y}=\bm{\gamma^2 \I & \0 \\ \0 & -\I}
=\bm{\tilde{\gamma} \I & \0 \\ \0 & -\I}
\end{equation}
and $0 < \tilde{\gamma} < \bar{\gamma}$. Note also that we can encode the hard constraint $\tilde{\gamma} < \bar{\gamma}$ as a soft constraint by including the $\tilde{\gamma}$ term in the objective function to be optimized in this synthesis.

Apart from the said $L_2$-gain $\gamma$ related term $\tilde{\gamma}$, in this objective function, we also propose to include a term analogous to $\sum_{i,j\in\N_N} c_{ij}\Vert K_{ij}\Vert_1$ where $C=[c_{ij}]_{i,j\in\N_N}\in\R^{N\times N}$ is a prespecified cost matrix corresponding to the inter-vehicle communication topology in the platoon. For example, if $N=4$ and the platoon requires a uniform separation in distance among the vehicles, to penalize inter-vehicle communications based on the length of each communication channel's length, one can use the cost matrix:
$$
C \triangleq \bm{0&1&2&3\\1&0&1&2\\2&1&0&1\\3&2&1&0}.
$$
It is worth noting that here we use an objective function containing 1-norms of the design variables (i.e., $\{\Vert K_{ij}\Vert_1:i,j\in\N_N\}$ as they are well-known to induce sparsity in the solution \cite{Sun2020,Tibshirani1996} -  leading to a sparse communication topology solution that is easy to implement and maintain.  

Taking these concerns into account, the following theorem formulates the centralized controller and topology synthesis problem as an LMI problem that can be solved using standard convex optimization toolboxes \cite{Boyd1994,Lofberg2004}.

\begin{theorem}\label{Th:CentralizedTopologyDesign}
The closed-loop error dynamics of the platoon (shown in Fig. \ref{Fig:PlatoonProblem}) can be made finite-gain $L_2$-stable with some $L_2$-gain $\gamma$ (where $\tilde{\gamma} \triangleq \gamma^2 < \bar{\gamma}$) from disturbance input $w(t)$ to performance output $z(t)$, by using the interconnection matrix block $M_{\eta e}=[K_{ij}]_{i,j\in\N_N}$ \eqref{Eq:PlatoonInterconnection} synthesized via solving the LMI problem:
\begin{subequations}\label{Eq:Th:CentralizedTopologyDesign}
\begin{equation}
\begin{aligned}
\min_{Q,\gamma,\{p_i: i\in\N_N\}}& \sum_{i,j\in\N_N} c_{ij} \Vert Q_{ij} \Vert_1 + c_0 \tilde{\gamma}, \\
\mbox{Sub. to: }& p_i > 0, \forall i\in\N_N,\ \  0 < \tilde{\gamma} < \bar{\gamma} \ \ \mbox{ and }
\end{aligned}
\end{equation}
\begin{equation}\label{Eq:Th:CentralizedTopologyDesignMain}
\scriptsize \bm{
\textbf{X}_p^{11} & \0 & Q & \textbf{X}_p^{11} \\
\0 & \I & \I & \0\\ 
Q^\T & \I & -Q^\T \textbf{X}^{12}-\textbf{X}^{21}Q-\textbf{X}_p^{22} & -\textbf{X}^{21}\textbf{X}_p^{11} \\
\textbf{X}_p^{11} & \0 & -\textbf{X}_p^{11} \textbf{X}^{12}&  \tilde{\gamma} \I
} \normalsize >0
\end{equation}
\end{subequations}
where $c_0>0$ is a prespecified constant, $Q\triangleq[Q_{ij}]_{i,j\in\N_N}$ shares the same structure as $M_{\eta e}$, $\textbf{X}^{12} \triangleq \diag(-\frac{1}{2\nu_i}\I:i\in\N_N)$, $\textbf{X}^{21} \triangleq (\textbf{X}^{12})^\T$,
$\textbf{X}_p^{11} \triangleq \diag(-p_i\nu_i\I:i\in\N_N)$, 
$\textbf{X}_p^{22} \triangleq \diag(-p_i\rho_i\I:i\in\N_N)$, and 
$M_{\eta e} \triangleq (\textbf{X}_p^{11})^{-1} Q$.
\end{theorem}

\begin{proof}
The proof follows from applying the subsystem dissipativity properties assumed in \eqref{Eq:SubsystemDissipativity} and the networked system dissipativity properties considered in \eqref{Eq:NetworkedSystemDissipativity} in the interconnection topology synthesis result stated in Prop. \ref{Pr:NSC2Synthesis}.
\end{proof}

From the preceding discussion, it is clear that the interconnection matrix block $M_{\eta e}=[K_{ij}]_{i,j\in\N_N}$ synthesized in Th. \ref{Th:CentralizedTopologyDesign} provides both the vehicle (global) controller gains and the inter-vehicle communication topology for the platoon that optimally attenuates the disturbance effects on the performance while also minimizing a measure of total communication cost.   

For completeness, here we provide the direct relationship between the synthesized interconnection matrix block $[K_{ij}]_{i,j\in\N_N}$ in Th. \ref{Th:CentralizedTopologyDesign} and the individual vehicle (global) controller gains required in \eqref{Eq:LocationError} or \eqref{Eq:PlatoonFeedbackControl2}. In particular, under the error dynamics formulations I and II, the off-diagonal elements of $[K_{ij}]_{i,j\in\N_N}$ are, respectively,
\begin{equation*}\label{Eq:ControllerGainsOffDiagonal}
    K_{ij} = \bm{0&0&0\\0&0&-\bar{k}_{ij}\\0&0&0} \mbox{ and } K_{ij} = \bm{0&0&0\\0&0&0\\-l_{ij}^x&-l_{ij}^v&-l_{ij}^a},
\end{equation*}
for all $i\in\N_N, j\in\N_N\backslash\{i\}$. The diagonal elements are
\begin{equation}\label{Eq:ControllerGainsDiagonal}
    K_{ii} = K_{i0} - \sum_{j\in\N_N\backslash\{i\}} K_{ij},
\end{equation}
for all $i\in\N_N$, where each $K_{i0},\,i\in\N_N$, under the error dynamics formulations I and II are, respectively,
\begin{equation*}\label{Eq:ControllerGainsDiagonalLeader}
    K_{i0} = \bm{0&0&0\\0&0&\bar{k}_{i0}-1\\l_{ii}^x&l_{ii}^v&l_{ii}^a} \mbox{ and } K_{i0} = \bm{0&0&0\\0&0&0\\l_{ii}^x&l_{ii}^v&l_{ii}^a}.
\end{equation*}

\subsection{Necessary Conditions on Subsystem Passivity Indices}
Based on \eqref{Eq:Th:CentralizedTopologyDesignMain} in Th. \ref{Th:CentralizedTopologyDesign} (particularly, see the terms involving $\textbf{X}_p^{11}$, $\textbf{X}_p^{22}$, $\textbf{X}^{12}$ and $\textbf{X}^{21}$), it is clear that both the feasibility and the effectiveness of the global controller design hinges on the used subsystem passivity indices $\{\nu_i,\rho_i:i\in\N_N\}$ \eqref{Eq:SubsystemDissipativity} assumed for the subsystems $\{\tilde{\Sigma}_i:i\in\N_N\}$ \eqref{Eq:PlatoonClosedLoop1}. 

However, based on Co. \ref{Co:LTISystemXDisspativation}, local controllers $\{\bar{L}_{ii}:i\in\N_N\}$ \eqref{Eq:VirtualControlInputLawForm} can be designed for the subsystems $\{\tilde{\Sigma}_i:i\in\N_N\}$ \eqref{Eq:PlatoonClosedLoop1} to achieve a certain set of subsystem passivity indices (e.g., the passivity indices that maximizes $\nu_i+\rho_i, i\in\N_N$), independently from the global controller design \eqref{Eq:Th:CentralizedTopologyDesign}.

As mentioned earlier, executing the local controller designs in such an independent manner can lead to an infeasible and/or ineffective global controller design. Therefore, local controller designs have to account for specific conditions required in global controller design \eqref{Eq:Th:CentralizedTopologyDesign}. A few of such conditions have been identified in the following lemma.

\begin{lemma}\label{Lm:CodesignConditions}
For the main LMI condition \eqref{Eq:Th:CentralizedTopologyDesignMain} in Th. \ref{Th:CentralizedTopologyDesign} to hold, it is necessary that, at each subsystem $\tilde{\Sigma}_i,i\in\N_N$ \eqref{Eq:PlatoonClosedLoop1}, for some scalar parameter $p_i>0$, the passivity indices $\nu_i$ and $\rho_i$ \eqref{Eq:SubsystemDissipativity} (let $\tilde{\rho}_i \triangleq \frac{1}{\rho_i}$) are such that the LMI problem: 
\begin{subequations}\label{Eq:Lm:CodesignConditions}
\begin{align}
\mbox{Find: }\ \ &\nu_i,\ \ \tilde{\rho}_i,\ \ \tilde{\gamma}_i, \quad 
\mbox{Sub. to: }\\
&-\frac{\tilde{\gamma}_i}{p_i}<\nu_i<0,   \label{Eq:Lm:CodesignConditions1}      \\
&0<\tilde{\rho_i}<p_i,       \label{Eq:Lm:CodesignConditions2}     \\
&0<\tilde{\rho_i}<\frac{4\tilde{\gamma}_i}{p_i},   \label{Eq:Lm:CodesignConditions3}
\end{align}
\end{subequations}
is feasible.
\end{lemma}
\begin{proof}
If \eqref{Eq:Th:CentralizedTopologyDesignMain} holds, there should exist some scalars $0<\tilde{\gamma}_i<\tilde{\gamma}, \forall i\in\N_N$. In other words, \eqref{Eq:Th:CentralizedTopologyDesignMain} holds if there exists $\tilde{\Gamma} \triangleq \diag(\tilde{\gamma}_i\I:i\in\N_N)$ such that $0<\tilde{\Gamma}<\tilde{\gamma}\I$ and 
\begin{equation}\label{Eq:Lm:CodesignConditionsStep1}
\scriptsize \bm{
\textbf{X}_p^{11} & \0 & Q & \textbf{X}_p^{11} \\
\0 & \I & \I & \0\\ 
Q^\T & \I & -Q^\T \textbf{X}^{12}-\textbf{X}^{21}Q-\textbf{X}_p^{22} & -\textbf{X}^{21}\textbf{X}_p^{11} \\
\textbf{X}_p^{11} & \0 & -\textbf{X}_p^{11} \textbf{X}^{12}&  \tilde{\Gamma}} \normalsize > 0.
\end{equation}
Now, for the LMI condition in \eqref{Eq:Lm:CodesignConditionsStep1} to hold, the following set of necessary conditions can be identified (by considering some specific sub(block)matrices of the block matrix in \eqref{Eq:Lm:CodesignConditionsStep1}):
\begin{subequations}
\begin{equation}\label{Eq:Lm:CodesignConditionsStep2_1}
\textbf{X}_p^{11} > 0    
\end{equation}
\begin{equation}\label{Eq:Lm:CodesignConditionsStep2_2}
-Q^\T \textbf{X}^{12}-\textbf{X}^{21}Q-\textbf{X}_p^{22}>0
\end{equation}
\begin{equation}\label{Eq:Lm:CodesignConditionsStep2_4}
\bm{\textbf{X}_p^{11} & \textbf{X}_p^{11}  \\
\textbf{X}_p^{11} & \tilde{\Gamma}}>0,
\end{equation}
\begin{equation}\label{Eq:Lm:CodesignConditionsStep2_5}
\bm{\I & \I  \\ \I & -Q^\T \textbf{X}^{12}-\textbf{X}^{21} Q-\textbf{X}_p^{22}}>0,    
\end{equation}
\begin{equation}\label{Eq:Lm:CodesignConditionsStep2_6}
\bm{-Q^\T \textbf{X}^{12}-\textbf{X}^{21} Q-\textbf{X}_p^{22} &  -\textbf{X}^{21}\textbf{X}_p^{11} \\ -\textbf{X}_p^{11} \textbf{X}^{12} & \tilde{\Gamma}}>0,    
\end{equation}
\end{subequations}

A necessary condition for \eqref{Eq:Lm:CodesignConditionsStep2_1} can be obtained as 
\begin{equation}\label{Eq:Lm:CodesignConditionsStep3_1}
\mbox{\eqref{Eq:Lm:CodesignConditionsStep2_1}}  \iff  -p_i\nu_i > 0 \iff  \nu_i < 0, \forall i\in\N_N.
\end{equation}
Considering the diagonal elements of the matrix in \eqref{Eq:Lm:CodesignConditionsStep2_2} and then using structure of the matrix $Q$ (particularly, using the fact that each $Q_{ii}$ block has zeros in its diagonal, for any $i\in\N_N$), a necessary condition for \eqref{Eq:Lm:CodesignConditionsStep2_2} can be obtained as 
\begin{align}
\mbox{\eqref{Eq:Lm:CodesignConditionsStep2_2}} \implies&\ 
-Q_{ii}^\T (-\frac{1}{2\nu_i}\I) - (-\frac{1}{2\nu_i}\I)Q_{ii} - (-p_i\rho_i\I) > 0 \nonumber \\
\label{Eq:Lm:CodesignConditionsStep3_2}
\implies&\ p_i\rho_i > 0 \iff \rho_i > 0, \forall i\in\N_N.
\end{align}

To obtain the necessary conditions for \eqref{Eq:Lm:CodesignConditionsStep2_4}-\eqref{Eq:Lm:CodesignConditionsStep2_6}, apart from the above two techniques, we also need to use the result in Pr. \ref{Pr:NetworkMatrixProperties} (Part 4). Consequently, we can obtain: 
\begin{align}
\mbox{\eqref{Eq:Lm:CodesignConditionsStep2_4}} \iff&\ \bm{-p_i\nu_i & -p_i\nu_i \\ -p_i\nu_i & \tilde{\gamma}_i } > 0 \iff -p_i\nu_i(\tilde{\gamma}_i + p_i\nu_i)  > 0 \nonumber \\ 
\label{Eq:Lm:CodesignConditionsStep3_4}
\iff&\ \nu_i > -\frac{\tilde{\gamma}_i}{p_i}, \forall i\in\N_N;
\end{align}
\begin{align}
\label{Eq:Lm:CodesignConditionsStep3_5}
\mbox{\eqref{Eq:Lm:CodesignConditionsStep2_5}} \implies&\ \bm{1 & 1 \\ 1 & p_i\rho_i} > 0 \iff \rho_i > \frac{1}{p_i}, \forall i\in\N_N; 
\end{align}
\begin{align}
\mbox{\eqref{Eq:Lm:CodesignConditionsStep2_6}} \implies&\ \bm{p_i\rho_i & -\frac{1}{2}p_i \\ -\frac{1}{2}p_i & \tilde{\gamma}_i} > 0 \iff p_i(\rho_i\tilde{\gamma}_i -\frac{1}{4}p_i) > 0 \nonumber \\
\label{Eq:Lm:CodesignConditionsStep3_6}
\iff&\ \rho_i > \frac{p_i}{4\tilde{\gamma}_i}, \forall i\in\N_N. 
\end{align}

Finally, by: (1) compiling the necessary conditions found in \eqref{Eq:Lm:CodesignConditionsStep3_1}-\eqref{Eq:Lm:CodesignConditionsStep3_6}, (2) using the change of variables $\tilde{\rho}_i \triangleq \frac{1}{\rho_i}, \forall i\in\N_N$, and (3) considering each $p_i, i\in\N_N$ as a predefined positive scalar parameter, we can obtain the set of LMI problems (each in variables $\nu_i,\ \tilde{\rho}_i,\ \tilde{\gamma}_i$, for any $i\in\N_N$) given in \eqref{Eq:Lm:CodesignConditions}. 
\end{proof}

Regarding the predefined parameters $\{p_i:i\in\N_N\}$ and the LMI variables $\{\tilde{\gamma}_i:i\in\N_N\}$ seen in the LMI problems \eqref{Eq:Lm:CodesignConditions} (and also their respective counterparts: the LMI variables $\{p_i:i\in\N_N\}$ and $\tilde{\gamma}$ seen in the LMI problem \eqref{Eq:Th:CentralizedTopologyDesign}), the following two respective remarks can be made.

\begin{remark} \label{Rm:CodesignVariables1}
In the LMI problems \eqref{Eq:Lm:CodesignConditions}, the reason for considering the parameters $\{p_i:i\in\N_N\}$ as ``predefined'' (as opposed to LMI variables) is that they appear non-linearly with the other LMI variables $\{\nu_i, \tilde{\rho}_i,\tilde{\gamma}_i: i\in\N_N\}$. Note, however, that, these predefined parameters $\{p_i:i\in\N_N\}$ are not required to be equal to the eventual values of the LMI variables $\{p_i:i\in\N_N\}$ seen in the LMI problem \eqref{Eq:Th:CentralizedTopologyDesign}. This is because, if for some set of predefined parameters $\{p_i:i\in\N_N\}$, the LMI problems \eqref{Eq:Lm:CodesignConditions} are feasible, then, it in itself can be thought of as a necessary condition for the LMI problem \eqref{Eq:Th:CentralizedTopologyDesign}. Note also that, these predefined parameters $\{p_i:i\in\N_N\}$ also enable the optimization of the overall local and global controller design process. Regardless, in practice, particularly for homogeneous networks, selecting $p_i = \frac{1}{N}, \forall i\in\N_N$ in the LMI problems \eqref{Eq:Lm:CodesignConditions} has rendered reasonable local and global controller designs.
\end{remark}

\begin{remark}  \label{Rm:CodesignVariables2}
Similar to the set $\{p_i:i\in\N_N\}$, one can also consider the set $\{\tilde{\gamma}_i: i\in\N_N\}$ in the LMI problems \eqref{Eq:Lm:CodesignConditions} as a set of ``predefined'' parameters. However, it is not necessary as $\{\tilde{\gamma}_i: i\in\N_N\}$ appear linearly with the other LMI variables $\{\nu_i, \tilde{\rho}_i: i\in\N_N\}$ in \eqref{Eq:Lm:CodesignConditions}. Further, under the same arguments made in Rm. \ref{Rm:CodesignVariables1}, note that these LMI variables $\{\tilde{\gamma}_i: i\in\N_N\}$ are not required to relate (in any way) to the LMI variable $\tilde{\gamma}$ seen in the LMI problem \eqref{Eq:Th:CentralizedTopologyDesign}. Note also that, minimizing these LMI variables $\tilde{\gamma}_i$ in the LMI problems \eqref{Eq:Lm:CodesignConditions} lead to higher passivity indices $\{\nu_i,\rho_i:i\in\N_N\}$ from \eqref{Eq:Lm:CodesignConditions}. And consequently, can lead to a lower/better $\tilde{\gamma}$ value in the LMI problem \eqref{Eq:Th:CentralizedTopologyDesign}. Therefore, these LMI variables $\{p_i:i\in\N_N\}$ also enable the optimization of the overall local and global controller design process.
\end{remark}

\subsection{Local Controller Sysnthesis}

Having established the fundamental result in Co. \ref{Co:LTISystemXDisspativation}, the subsystem (vehicle error dynamics) models in \eqref{Eq:PlatoonClosedLoop1}, and the necessary conditions on subsystem passivity indices in Lm. \ref{Lm:CodesignConditions}, we are now ready to present the local controller synthesis problem as an LMI problem formally in a theorem. 

\begin{theorem}\label{Th:LocalControllerDesign}
At each vehicle $\Sigma_i, i\in\N_N$, for some prespecified scalar parameter $p_i>0$, to make the corresponding closed-loop error dynamics (subsystem) $\tilde{\Sigma}_i$ \eqref{Eq:PlatoonClosedLoop1} IF-OFP($\nu_i,\rho_i$) (as assumed in \eqref{Eq:SubsystemDissipativity}) while also satisfying the necessary conditions identified in Lm. \ref{Lm:CodesignConditions}, the local controller gains $\bar{L}_{ii}$ \eqref{Eq:VirtualControlInputLawForm} should be synthesized via solving the LMI problem:
\begin{equation}\label{Eq:Th:LocalControllerDesign}
\begin{aligned}
\mbox{Find: }&\ \tilde{L}_{ii},\ P_i,\ \nu_i,\ \tilde{\rho}_i,\ \tilde{\gamma}_i, \\
\mbox{Sub. to: }&\ P_i > 0, \\
&\ 
\bm{\tilde{\rho}_i\I & P_i & \0 \\
P_i &-\mathcal{H}(AP_i + B \tilde{L}_{ii})& -I + \frac{1}{2}P_i\\
\0 & -I + \frac{1}{2}P_i & -\nu_i\I} > 0,\\ 
&\ -\frac{\tilde{\gamma}_i}{p_i}<\nu_i<0, \\
&\ 0<\tilde{\rho_i}< \min\left\{p_i,\frac{4\tilde{\gamma}_i}{p_i}\right\}
\end{aligned}
\end{equation}  
where $\rho_i \triangleq \frac{1}{\tilde{\rho}_i}$ and $\bar{L}_{ii} \triangleq \tilde{L}_{ii}P_i^{-1}$.
\end{theorem}
\begin{proof}
The proof follows from: (1) considering the vehicle error dynamics (subsystem) model identified in \eqref{Eq:PlatoonClosedLoop1} and then applying the LMI-based controller synthesis technique given in Co. \ref{Co:LTISystemXDisspativation} so as to enforce the subsystem passivity indices assumed in \eqref{Eq:SubsystemDissipativity}, and then (2) enforcing the necessary conditions on subsystem passivity indices identified in Lm. \ref{Lm:CodesignConditions} so as to ensure a feasible and effective global controller design at \eqref{Eq:Th:CentralizedTopologyDesign}. 
\end{proof}
 
To conclude this section, we summarize the three main steps necessary for designing local controllers, global controllers and communication topology (in a centralized manner, for the considered vehicular platoon) as: 

\noindent
\textbf{Step\,1:} Pre-specify the scalar parameters: $p_i>0,\forall i\in\N_N$;\\
\noindent
\textbf{Step\,2:} Synthesize local controllers using Th. \ref{Th:LocalControllerDesign};\\
\noindent
\textbf{Step\,3:} Syntesize global controllers and topology using Th. \ref{Th:CentralizedTopologyDesign}.

\begin{remark}
Note that, apart from executing the above three steps just once, they can also be executed in an iterative manner, because, \textbf{Step 3} may reveal a more suitable set of parameters $\{p_i:i\in\N_N\}$ to be used in a \textbf{Step 1} in a following iteration. The effectiveness and the convergence of such an iterative co-design scheme is the subject of on-going research and is considered beyond the scope of this paper. On the other hand, one can also use a bruteforce or non-linear optimization approach to find a candidate set of parameters $\{p_i:i\in\N_N\}$ that renders a lower objective function value at \textbf{Step 3}.     
\end{remark}

\begin{remark}
According to \eqref{Eq:PlatoonClosedLoop1}, subsystems $\{\tilde{\Sigma}_i:i\in\N_N\}$ are linear and homogeneous. This is mainly a consequence of the used feedback linearization step \eqref{Eq:LinearizingControlInput}. If needed, one can promote/retain heterogeneity in \eqref{Eq:PlatoonClosedLoop1} by modifying \eqref{Eq:LinearizingControlInput} to exclude some natural and helpful terms (e.g., $-\frac{1}{\tau_i}a_i$) from being canceled in the original vehicle dynamics \eqref{Eq:VehicleDynamics}. Similarly, we can also retain/promote some non-linearities in \eqref{Eq:PlatoonClosedLoop1} (from \eqref{Eq:VehicleDynamics}), particularly if they are helpful and lead to stronger dissipativity properties than those can be established via Th. \ref{Th:LocalControllerDesign}. In all, such modifications are helpful and can highlight the impact of the proposed dissipativity-based approach. 
\end{remark}



\section{Decentralize Design of the Platoon}\label{Sec:DecentralizedCoDesign}

In this section, we first show how the centralized controller and topology co-design problems formulated in Theorems \ref{Th:CentralizedTopologyDesign} and \ref{Th:LocalControllerDesign} can be executed in a decentralized and compositional manner over the platoon. Subsequently, exploiting the said compositionality features we detail how scenarios like merging and splitting that occur in vehicular platoons can be handled efficiently. Finally, we discuss the analysis and enforcement of the string stability properties of the resulting platoons.

\subsection{Decentralized Co-design of Controllers and Topology}\label{subsec:Decentralized_Codesign}

To decentrally implement the LMI problem \eqref{Eq:Th:CentralizedTopologyDesign}, we cannot directly apply the network matrices-based decentralization technique proposed in \cite{WelikalaP32022}. This is due to the dependence of the diagonal blocks in $M_{\eta e}=[K_{ij}]_{i,j\in\N_N}$ (synthesized in \eqref{Eq:Th:CentralizedTopologyDesign}) on its off-diagonal blocks (see \eqref{Eq:ControllerGainsDiagonal}), and the presence of the global constraint $\tilde{\gamma} < \bar{\gamma}$. To overcome these two challenges, we propose to exploit the term $K_{i0}$ included in each diagonal term $K_{ii}$ in $M_{\eta e}$, and the concept of diagonal dominance, respectively. The following theorem provides the proposed decentralized implementation of the LMI problems \eqref{Eq:Th:CentralizedTopologyDesign} and \eqref{Eq:Th:LocalControllerDesign} that can be used to decentrally and compositionally synthesize the controllers and the topology for the platoon.

\begin{theorem}\label{Th:DecentralizedTopologyDesign}
The closed-loop error dynamics of the platoon (shown in Fig. \ref{Fig:PlatoonProblem}) can be made finite-gain $L_2$-stable with some $L_2$-gain $\gamma$ (where  $\tilde{\gamma} \triangleq \gamma^2 < \bar{\gamma}$ from disturbance input $w(t)$ to performance output $z(t)$, if at each vehicle $\Sigma_i, i\in\N_N$: 
(1) the local controller gains $\bar{L}_{ii}$ \eqref{Eq:PlatoonClosedLoop1} are designed using the local LMI problem \eqref{Eq:Th:LocalControllerDesign}, (2) the interconnection/global controller gain blocks $\{K^i\}$ are designed using the local LMI problem:
\begin{equation}\label{Eq:Th:DecentralizedTopologyDesign}
\begin{aligned}
\min_{\{Q^i\},\ \hat{\gamma}_i,\ p_i}& \sum_{j\in\N_{i-1}}c_{ij}\Vert Q_{ij} \Vert_1 + c_{ji} \Vert Q_{ji} \Vert_1 + c_{0i}\hat{\gamma}_i + c_{i}\vert \hat{\gamma}_i-\tilde{\gamma}_i\vert\\ 
\mbox{Sub. to: }& p_i>0,\ \hat{\gamma}_i < \bar{\gamma},\ \tilde{W}_{ii}>0,
\end{aligned}
\end{equation}
where $\tilde{\gamma}_i$ is from \eqref{Eq:Th:LocalControllerDesign} (obtained in Step 1), $\tilde{W}_{ii}$ is from \eqref{Eq:Pr:MainProposition1} when enforcing $W=[W_{ij}]_{i,j\in \N_N} > 0$ with 
\begin{equation*}
W_{ij} \triangleq 
\scriptsize \bm{
\mathsf{e}_{ij}V_p^{ii} & \0 & Q_{ij} & \mathsf{e}_{ij}V_p^{ii} \\
\0 & \mathsf{e}_{ij}\I & \mathsf{e}_{ij}\I & \0\\ 
Q_{ji}^\T & \mathsf{e}_{ij}\I & -Q_{ji}^\T S_{jj}-S_{ii}Q_{ij}-\mathsf{e}_{ij}R_p^{ii} & -\mathsf{e}_{ij}S_{ii}V_p^{ii} \\
\mathsf{e}_{ij}V_p^{ii} & \0 & -\mathsf{e}_{ij}V_p^{ii} S_{jj} &  \hat{\gamma}_i\mathsf{e}_{ij} \I
},
\end{equation*}
$V_p^{ii} \triangleq -p_i\nu_i \I$,
$R_p^{ii} \triangleq -p_i\rho_i \I$,
$S_{ii} \triangleq -\frac{1}{2\nu_i}\I$ and blocks $\{K^{i}\}$ are determined by $K_{ij} = (V_p^{ii})^{-1}Q_{ij}$, and 
(3) the update: 
\begin{equation}\label{Eq:Th:DecentralizedTopologyDesign2}
 K_{j0} := K_{j0} + K_{ji}.   
\end{equation}
(using the obtained blocks $\{K_{ji}:j\in\N_{i-1}\}$ in Step 2) is requested at each prior neighboring vehicle $\Sigma_j, j\in\N_{i-1}\cap \F_i$. 
\end{theorem}
\begin{proof}
According to Prop. \ref{Pr:MainProposition}, enforcing $\tilde{W}_{ii}>0$ at each vehicle $\Sigma_i,i\in\N_N$ is equivalent to enforcing $W=[W_{ij}]_{i,j\in\N_N} > 0$ for the entire platoon $\Sigma$. It is easy to see that this $W=\text{BEW}(\bar{W})$, where $\bar{W}$ is (also using some notations in Th. \ref{Th:CentralizedTopologyDesign}):
\begin{equation}
\normalsize \bar{W}  \triangleq 
    \scriptsize \bm{
		\textbf{X}_p^{11} & \0 & Q & \textbf{X}_p^{11} \\
		\0 & \I & \I & \0\\ 
		Q^\T & \I & -Q^\T \textbf{X}^{12}-\textbf{X}^{21}Q-\textbf{X}_p^{22} & -\textbf{X}^{21}\textbf{X}_p^{11} \\
		\textbf{X}_p^{11} & \0 & -\textbf{X}_p^{11} \textbf{X}^{12}&  \Gamma
	}\normalsize 
\end{equation}
with $\Gamma \triangleq \diag([\hat{\gamma}_i\I]_{i\in\N_N})$. Let us define $\tilde{\gamma} \triangleq \max_{i\in\N_N}\hat{\gamma}_i$, which satisfies $\tilde{\gamma} < \bar{\gamma}$ due to local enforcements of $\hat{\gamma}_i < \bar{\gamma}$ at each vehicle $\Sigma_i, i\in\N_N$, and $\Gamma < \tilde{\gamma} \I$ due to the diagonal dominance concept. Since $\{\tilde{W}_{ii}>0:\forall i\in\N_N\} \iff W = \text{BEW}(\bar{W}) > 0 \iff \bar{W}>0 \implies  \mbox{\eqref{Eq:Th:CentralizedTopologyDesignMain}}$ (respectively from Prop. \ref{Pr:MainProposition}, Prop. \ref{Pr:NetworkMatrixProperties} and the diagonal dominance concept), and $\{\hat{\gamma}_i < \bar{\gamma}:\forall i\in\N_N\} \implies \tilde{\gamma} < \bar{\gamma}$, it is clear that any solution that satisfies the decentralized set of LMI conditions in \eqref{Eq:Th:DecentralizedTopologyDesign} will also satisfy the centralized set of LMI conditions in \eqref{Eq:Th:CentralizedTopologyDesign}. 

However, there is one remaining concern due to the special dependence explained in \eqref{Eq:ControllerGainsDiagonal}. Note that the proposed decentralized process \eqref{Eq:Th:DecentralizedTopologyDesign} derives the interconnection/controller gain blocks iteratively, i.e., at each vehicle/iteration $\Sigma_i,i\in\N_N$, the set of blocks 
$\{K^i\}\triangleq \{K_{ii}\}\cup\{K_{ij}:j\in\N_{i-1}\}\cup\{K_{ji}:j\in\N_{i-1}\}$ is derived. In such an iteration (say at vehicle $\Sigma_i,i\in\N_N$), due to \eqref{Eq:ControllerGainsDiagonal}, each derived matrix $K_{ji}, j\in\N_{i-1}$ will affect the matrix $K_{jj}$ derived previously at the prior neighboring vehicle $\Sigma_j$ (of vehicle $\Sigma_i$). Explicitly,
\begin{equation}
    K_{jj}^{\text{New}} = \left(K_{j0} - \sum_{l<i, l \neq j} K_{jl}\right) - K_{ji} = K_{jj}^{\text{Old}} - K_{ji}. 
\end{equation}
Since the proposed decentralization (i.e., the application of Prop. \ref{Pr:MainProposition}) is valid only if $W$ is a network matrix, we require both $Q$, and by extension, $[K_{ij}]_{i,j\in\N_N}=[(V_p^{ii})^{-1}Q_{ij}]_{i,j\in\N_N}$, to be network matrices (see Prop. \ref{Pr:NetworkMatrixProperties}). The latter requirement holds only if we constrain $K_{jj}^{\text{New}}=K_{jj}^{\text{Old}}$. For this, we require the updates given in \eqref{Eq:Th:DecentralizedTopologyDesign2}. This completes the proof.
\end{proof}

\begin{remark}
The design of local controllers given in Th. \ref{Th:LocalControllerDesign}, due to its decentralized form, seamlessly fits in with the overall decentralized solution given in Th. \ref{Th:DecentralizedTopologyDesign}. 
Consequently, the local controller design (under Step 1 of Th. \ref{Th:DecentralizedTopologyDesign}) enforces a favorable set of local passivity indices for the subsequent decentralized global controller design (under Step 2 of Th. \ref{Th:DecentralizedTopologyDesign}). 
Note, however, that, the objective function used in this decentralized global controller design \eqref{Eq:Th:DecentralizedTopologyDesign} now includes a term $\vert \hat{\gamma}_i-\tilde{\gamma}_i\vert$ that penalizes the mismatch between the local and global controller designs. The role of this term is to prevent such a mismatch from growing over the decentralized iterations in Th. \ref{Th:DecentralizedTopologyDesign}.
\end{remark}

\begin{remark}
\label{Rm:OptimalityAndConservativeness}
Even though we can use Th. \ref{Th:DecentralizedTopologyDesign} to decentrally synthesize the controllers and topology so that a global constraint like $\gamma^2 \leq \bar{\gamma}$ holds for the platoon, its optimality is not guaranteed when compared to the centrally synthesized controllers and topology given by Th. \ref{Th:CentralizedTopologyDesign}. However, the LMI conditions in the decentralized implementation \eqref{Eq:Th:DecentralizedTopologyDesign} are more flexible than those in the centralized implementation \eqref{Eq:Th:CentralizedTopologyDesign} due to the involved additional design variables in \eqref{Eq:Th:DecentralizedTopologyDesign}. Regardless, the decentrally synthesized controllers and topology can be more conservative compared to their centralized counterparts due to the update steps \eqref{Eq:Th:DecentralizedTopologyDesign2}. In the ongoing research, we attempt to manipulate the cost function coefficients $\{c_{ij}: i,j\in\N_N\}\cup\{(c_{0i},c_i):i\in\N_N\}$ in \eqref{Eq:Th:DecentralizedTopologyDesign} to mitigate such optimality and conservativeness concerns. 
\end{remark}

\begin{remark}
\label{Rm:CustomStringStability}
One advantage of the proposed decentralized design \eqref{Eq:Th:DecentralizedTopologyDesign} is that it can be used to encode custom string stability conditions. For example, let us consider a custom string stability condition for the platoon error dynamics as: 
\begin{equation*}
\begin{aligned}
    \sup_{[w_j(\cdot)]_{j\in\N_i}}  \frac{\Vert [z_j(\cdot)]_{j\in\N_i} \Vert^2}{\Vert [w_j(\cdot)]_{j\in\N_i} \Vert^2} <&  \sup_{[w_j(\cdot)]_{j\in\N_{i-1}}}  \frac{\Vert [z_j(\cdot)]_{j\in\N_{i-1}} \Vert^2}{\Vert [w_j(\cdot)]_{j\in\N_{i-1}} \Vert^2},
\end{aligned}
\end{equation*}
for all $i\in\N_N$ (with the right hand side for $i=1$ being the prespecified constant $\bar{\gamma}$). It is easy to see that this condition can be encoded in the proposed decentralized design \eqref{Eq:Th:DecentralizedTopologyDesign} by replacing its local constraint $\hat{\gamma}_i < \bar{\gamma}$ by an alternative local constraint of the form $\hat{\gamma}_i < \bar{\gamma}_i \triangleq \min_{j\in\N_{i-1}}\hat{\gamma}_j$. 
\end{remark}

\vspace{-5mm}
\subsection{Vehicular Merging and Splitting}

The proposed decentralized scheme in Th. \ref{Th:DecentralizedTopologyDesign} is ``compositional'' due to two main reasons: (1) each of its iterations (decentralized steps) is independent of the total number of vehicles in the platoon (i.e., $N$), and (2) at any iteration, the designed controllers do not affect the already designed controllers at previous iterations (except for the minor update \eqref{Eq:Th:DecentralizedTopologyDesign2}). This compositionality property allows one to conveniently and efficiently append new vehicles to an existing platoon, and thus, it is extremely helpful in handling scenarios where multiple platoons are merged in sequence.

Moreover, the proposed decentralized scheme in Th. \ref{Th:DecentralizedTopologyDesign} 
can be executed in any arbitrary order. In other words, we can start the decentralized design procedures from any arbitrary vehicle (not necessarily from $\Sigma_1$) and continue until all the vehicles in the platoon are covered. While the optimality and the conservativeness of the resulting decentralized designs may vary depending on the followed order of the vehicles (see also Rm. \ref{Rm:OptimalityAndConservativeness}), a unique implication of this quality is that it can be used to conveniently and efficiently add an external vehicle into an existing platoon in the middle (not necessarily at the end) of that platoon. Consequently, this flexibility can be used to handle scenarios where multiple platoons are merged in parallel (i.e., not necessarily in sequence).

On the other hand, the compositionality property of the proposed decentralized designs can also be exploited to split an existing platoon into multiple platoons. For example, consider a situation where the decentralized design procedure is executed sequentially from $\Sigma_1$ to $\Sigma_N$, in that order. Now, if the local design problem \eqref{Eq:Th:DecentralizedTopologyDesign} 
fails to find a feasible solution at some vehicle $\Sigma_i,i\in\N_N\backslash\{1,N\}$, we can start a new platoon from that vehicle $\Sigma_i$ and resume the design procedure there on-wards. This splitting strategy is particularly useful when enforcing additional stricter constraints like the string stability conditions mentioned in Rm. \ref{Rm:CustomStringStability} (also discussed in detail next).

\subsection{String Stability Properties of the Platoon}

We conclude this section by providing some formal string stability results (recall Defs. \ref{def:l2_weak_ss} and \ref{def:DSS}) for the closed-loop error dynamics of the platoon (shown in Fig. \ref{Fig:PlatoonProblem}). In particular, these string stability results stem from the local passivity and global finite-gain $L_2$ stability properties enforced for the platoon using local and global controller designs, respectively (e.g., via Th. \ref{Th:DecentralizedTopologyDesign}). In the following theorem, we first establish the $L_2$ weak string stability (defined in Def. \ref{def:l2_weak_ss}) for the closed-loop error dynamics of the platoon.

\begin{theorem}\label{Th:l2_wss}
If the closed-loop error dynamics of the platoon is finite-gain $L_2$-stable with some $L_2$ gain $\gamma$ such that $\gamma^2 < \bar{\gamma}$ as in Theorems \ref{Th:CentralizedTopologyDesign} or \ref{Th:DecentralizedTopologyDesign}, then it is $L_2$ weakly string stable.
\end{theorem}
\begin{proof}
The differential dissipativity inequality corresponding to the established $L_2$-stability property takes the form:
\begin{equation}\label{Eq:Th:l2_wssStep1}
\dot{V}(e(t),e^*) \leq \bar{\gamma}^2 \vert w(t) \vert^2 - \vert z(t) \vert^2,   \quad \forall t\geq 0,
\end{equation}
where $e(t), w(t)$ and $z(t)$ are the error state, disturbance input and performance output of the closed-loop error dynamics of the platoon, respectively, and $e^* = \0$ (see also Def. \ref{Def:EID} and Rm. \ref{Rm:X-DissipativityVersions}). Assuming the initial state as $e(0)=\0$, upon integrating \eqref{Eq:Th:l2_wssStep1} and using the fact that $V(\0,\0)=0$, we can obtain  
\begin{equation}\label{Eq:Th:l2_wssStep2}
V(e(\tau),\0) \leq \int_0^\tau \bar{\gamma}^2 \vert w(t) \vert^2 - \vert z(t) \vert^2 dt, \quad \forall \tau \geq 0.
\end{equation}
Since $0 \leq V(e(\tau),\0)$, for $\tau=\infty$, \eqref{Eq:Th:l2_wssStep2} implies $0 \leq \bar{\gamma}^2 \Vert w \Vert^2 - \Vert z \Vert^2$. Now, using the fact that $z(t) \triangleq e(t)$, we get
\begin{equation}
\Vert e \Vert \leq \bar{\gamma} \Vert w \Vert, 
\end{equation}
which implies that $L_2$ weakly string stability condition \eqref{Eq:l2_wss_condition} in Def. \ref{def:l2_weak_ss} (with $\alpha_1(s) \triangleq \bar{\gamma}s$ and any arbitrary class $\mathcal{K}$ function $\alpha_2(s)$ independent of $N$) for the considered closed-loop error dynamics of the platoon. This completes the proof.
\end{proof}

\vspace{-1mm}

\section{Numerical Results}\label{Sec:NumericalResults}

In this section, we verify our proposed co-design framework in the previous sections through numerical simulations. We study two cases of the platooning centralized and decentralized co-design with $L_2$ weak string stability, respectively. 
For each case, we consider a vehicular platoon containing ten vehicles; the first vehicle is viewed as the leader, and the rest are the followers. Without loss of generality, we assume the platoon to be homogeneous and each vehicle in \eqref{Eq:VehicleDynamics} is with the parameters $m_i=1500\ \mbox{kg}$, $\tau_i=0.25\ \mbox{s}$, $A_{f,i}=2.2\ \mbox{m}^2$, $\rho=0.78\ \mbox{kg/m}^3$, $C_{d,i}=0.35$ and $C_{r,i}=0.067$, for any $i\in\N_{9}$. Besides, the desired separating distances between two successive vehicles and the length of a vehicle are selected as $\delta_{di}=5\ \mbox{m}$ and $L_i=2.5\ \mbox{m}$, respectively. Furthermore, the external disturbances are assumed to be Gaussian random noise with the mean value $w_m\sim\mbox{U}_{3\times 1}(-0.5,0.5)$ and the variance $w_v$ is normal distribution with standard deviation of the distribution $\mbox{U}_{3\times 1}(0,0.1)$ ($\mbox{U}_{3\times 1}$ stands for a $3\times 1$ vector with each element being uniform distribution). 
The initial states of all vehicles in the platoon are assumed to be $x_i(0)=x_0-x_{im}-x_{iv}$, for all $i\in\N_{9}$, where $x_0=0$ is the position of the leader, $x_{im}=5$ and $x_{iv}\sim\mbox{U}(-1,1)$ are the mean and variance of the separation between two successive vehicles, respectively.

Simulation results are generated by a simulator developed in MATLAB\footnote{Publicly available in \url{https://github.com/NDzsong2/Longitudinal-Vehicular-Platoon-Simulator.git}}. Considered simulation configurations can be seen in Figs. \ref{Fig:L2_WSS_CentralizedTopology} and \ref{Fig:L2_WSS_DecentralizedTopology}, where the simulation time, cost, normed position, velocity and acceleration tracking errors are presented. Certain state profiles and tracking errors are also plotted in Figs. \ref{Fig:L2_WSS_PositionTracking_CentralizedPlatoon}-\ref{Fig:L2_WSS_VelocityTrackingErrors_CentralizedPlatoon} and Figs. \ref{Fig:L2_WSS_PositionTracking_DecentralizedPlatoon}-\ref{Fig:L2_WSS_VelocityTrackingErrors_DecentralizedPlatoon}. For the simulation of both cases, we assume the leader's reference velocity to be a piecewise linear function (as shown in Fig. \ref{Fig:leader_reference_velocity}):
\begin{equation}
v_1(t): \left \{\begin{array}{ll}
15t\ \mbox{m/s}, & 0\leq t\leq 2\ \mbox{s}\\
(5t+20)\ \mbox{m/s}, & 2\leq t\leq 4\ \mbox{s}\\
40\ \mbox{m/s}, & 4\leq t\leq 6\ \mbox{s}\\
(-10t+100)\ \mbox{m/s}, & 6\leq t\leq 8\ \mbox{s}\\
20\ \mbox{m/s}. & 8\leq t\leq 10\ \mbox{s}
\end{array}\right.
\end{equation}
\begin{figure}[!h]
    \centering
    \includegraphics[width=3.5in]{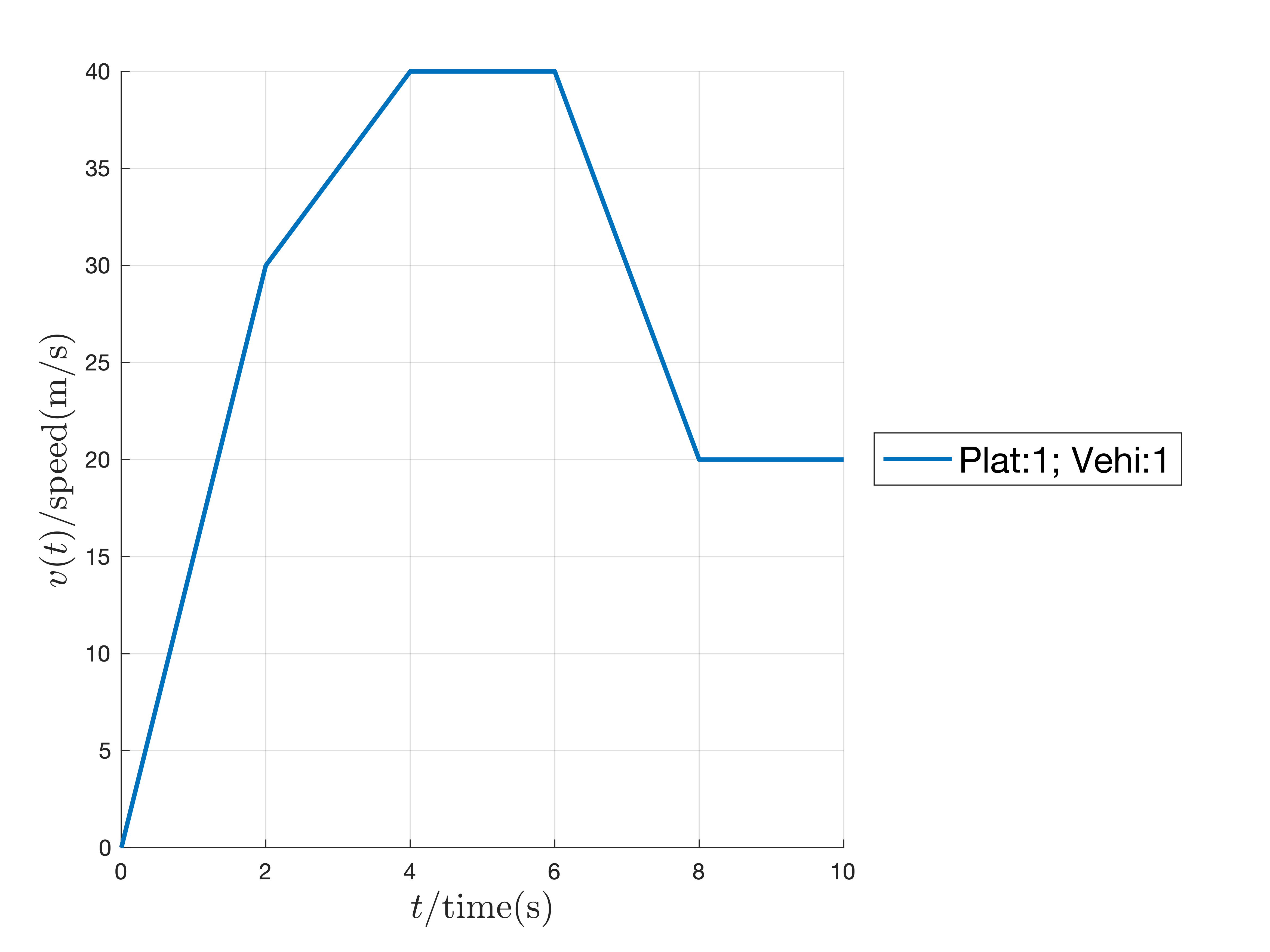}
    \caption{Reference velocity of the leader. }
    \label{Fig:leader_reference_velocity}
\end{figure}

\subsection{Case 1: $L_2$ Weak String Stability-based Centralized and Decentralized Platooning Co-Design}

We first verify the centralized version of the proposed co-design method, i.e., the program \eqref{Eq:Th:CentralizedTopologyDesign} in Theorem \ref{Th:CentralizedTopologyDesign} (with the program \eqref{Eq:Th:LocalControllerDesign} in Theorem \ref{Th:LocalControllerDesign}).
Based on the above platoon parameters, our simulation results are outlined in Figs. \ref{Fig:L2_WSS_CentralizedTopology}-\ref{Fig:L2_WSS_VelocityTrackingErrors_CentralizedPlatoon}. 
Specifically, the platoon under the synthesized centralized communication topology is shown in the following Fig. \ref{Fig:L2_WSS_CentralizedTopology}, where the communication links are dense between further neighbors and nearly all the followers maintain the connection with the leader. This results from the trade-off between the communication cost and the stability and dissipativity behaviors. For example, the communication links from the successive vehicles are beneficial for the stability and dissipativity of the platoon. Still, the density of such links will significantly increase the communication cost. Therefore, links from further neighboring vehicles and those from the leader are kept, meaning these links are more efficient than the nearest neighbors in terms of cost and performance.

The obtained (optimal) $L_2$-gain from the centralized co-design is $\gamma = 2.5093$.
To quantify the tracking control result and the string stability, the position and velocity tracking of the platoon under the centrally synthesized controller are further plotted in Figs. \ref{Fig:L2_WSS_PositionTracking_CentralizedPlatoon}-\ref{Fig:L2_WSS_VelocityTrackingErrors_CentralizedPlatoon}. From these plots, it is readily observed that the position and velocity tracking of each vehicle can be well achieved, despite of some offsets at the steady states. 
This is because of the existence of noise in the model and the method we proposed is inherently dissipativity-based without any compensation. 
Specifically, the real-time position of each vehicle is smoothly evolved and no collisions (or intersections between two curves of two vehicles) occur. Besides, the real-time tracking errors of positions and velocities are bounded and non-increasing along the platoon.
In fact, the resulting controller satisfies not only the $L_2$ weakly string stability (via Th. \ref{Th:l2_wss}) but also the so-called Disturbance String Stability (DSS) \cite{besselink2017string}, since the weak coupling condition $\max\{\sum_{i\in\N_9}|P_iK_{ij}|\}=\max\{2.9212e{-8},1.8539e{-8},1.0814e{-8},1.1025e{-8},
6.27\text{-}
\\
60e{-9},5.0243e{-9},1.2251e{-8},2.3824e{-8},2.1363e{-8}\}<1$ holds for all $i\in\N_9$.

\begin{figure}[!h]
    \centering
    \includegraphics[width=3in]{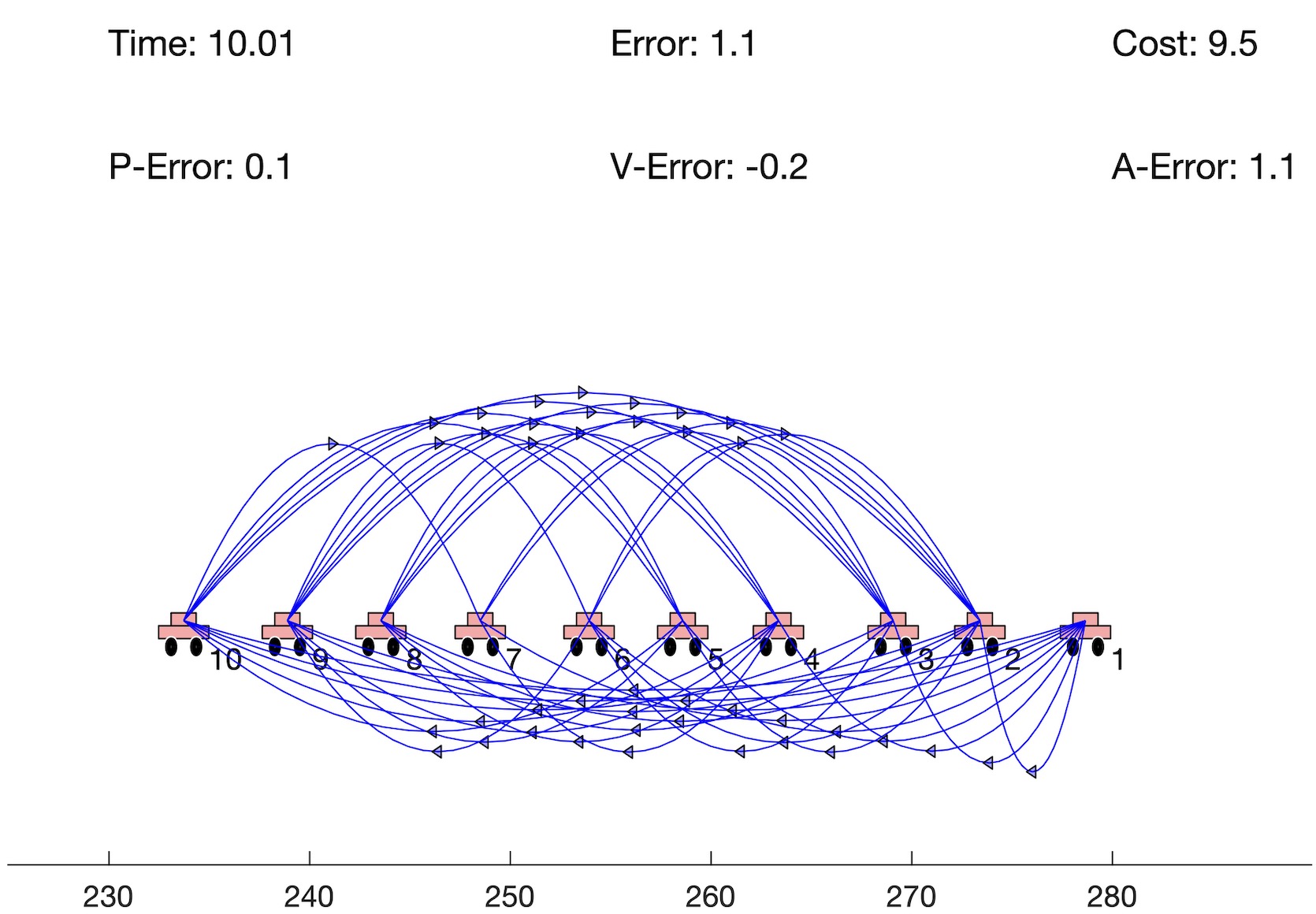}
    \vspace{-3mm}
    \caption{The platoon under the centralized controller (topology). Directed arrows represent the interconnection topology, where the upper ones are the communication edges from the follower to the predecessor and vice versa. }
    \label{Fig:L2_WSS_CentralizedTopology}
\end{figure}

\begin{figure}[!h]
    \centering
    \includegraphics[width=3in]{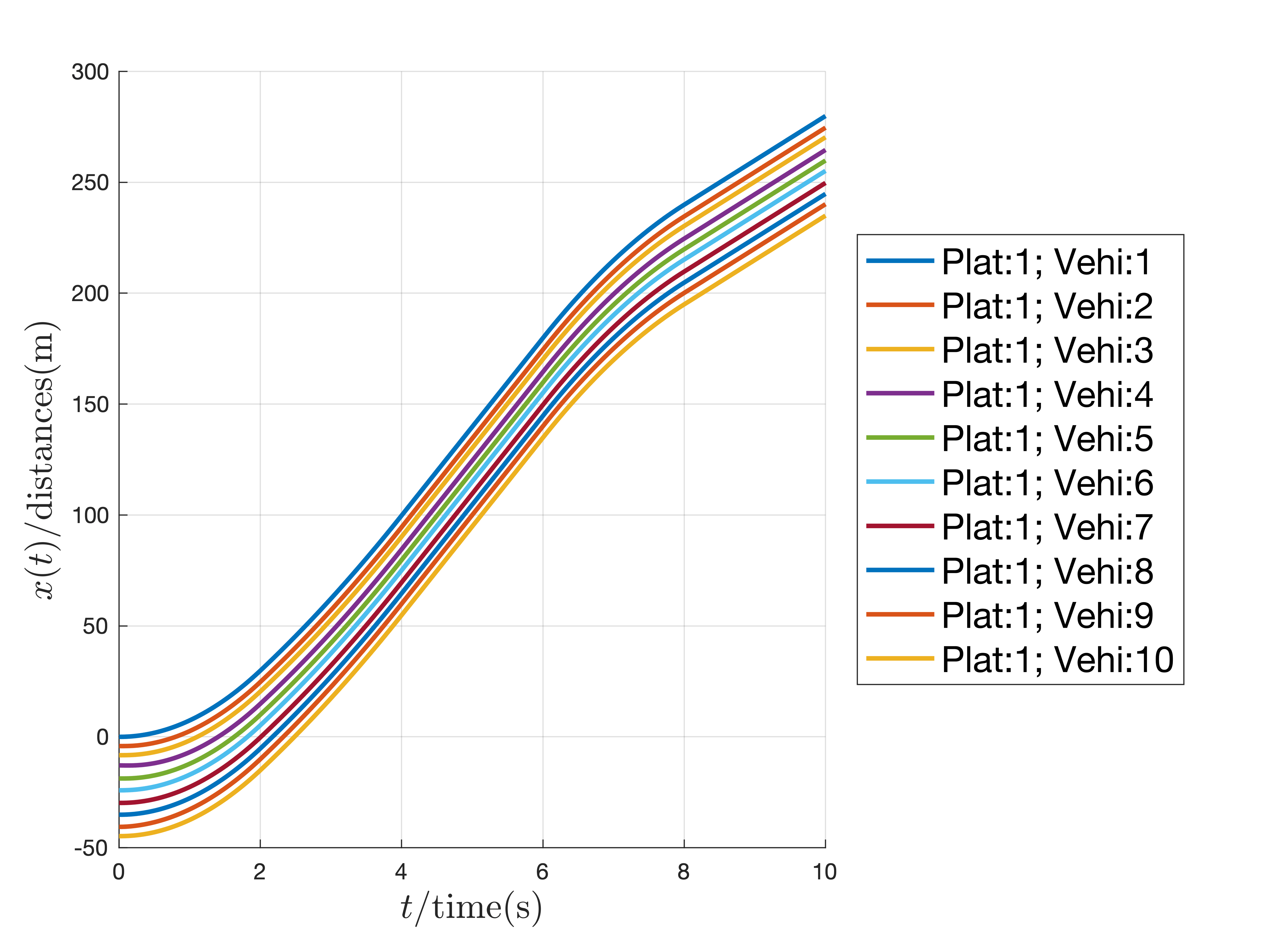}
    \vspace{-3mm}
    \caption{Position tracking under centralized co-design.}
    \label{Fig:L2_WSS_PositionTracking_CentralizedPlatoon}
\end{figure}

\begin{figure}[!h]
    \centering
    \includegraphics[width=3in]{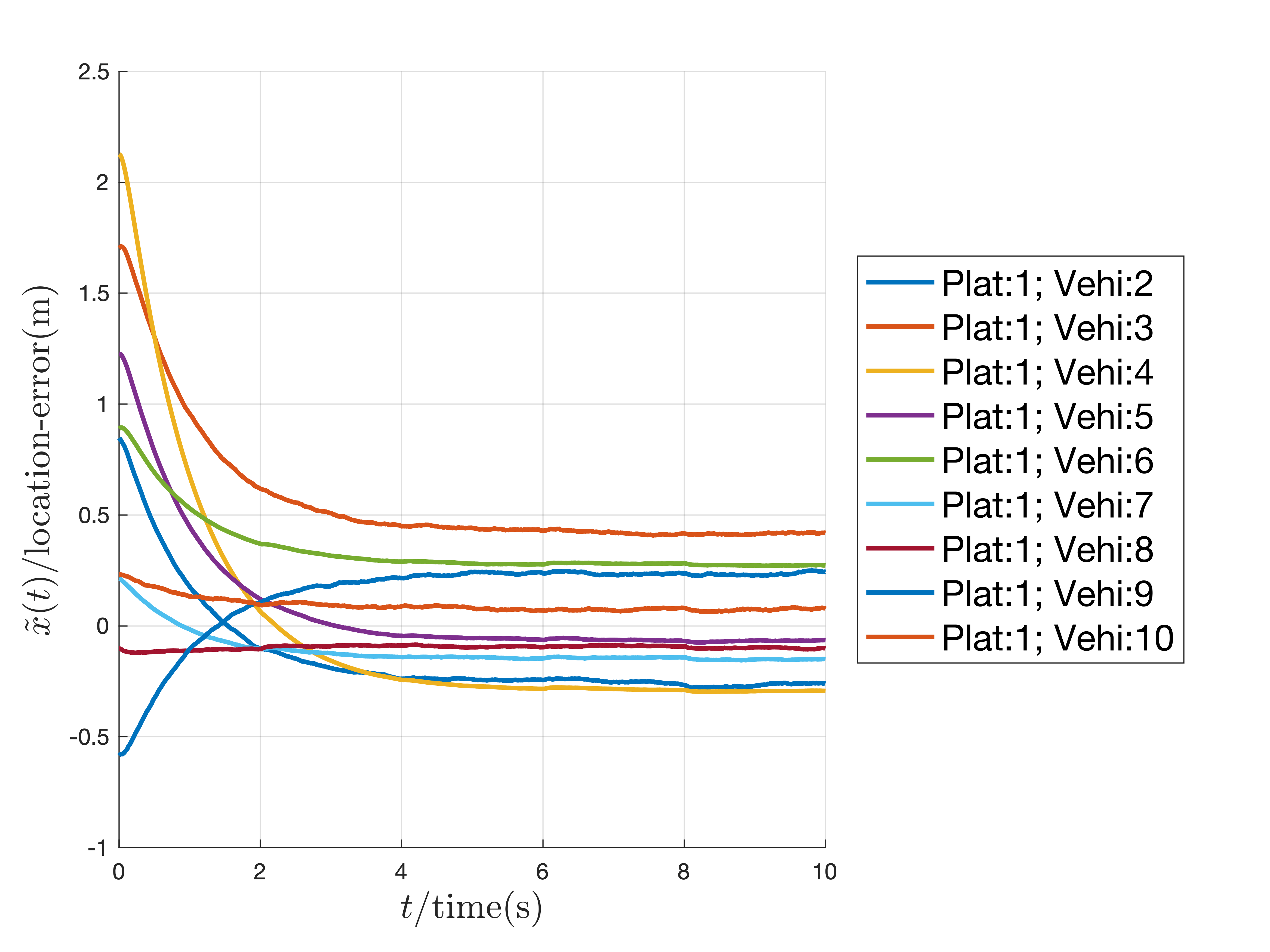}
    \vspace{-3mm}
    \caption{Position tracking errors under centralized co-design.}
    \label{Fig:L2_WSS_PositionTrackingErrors_DecentralizedPlatoon}
\end{figure}

\begin{figure}[!h]
    \centering
    \includegraphics[width=3in]{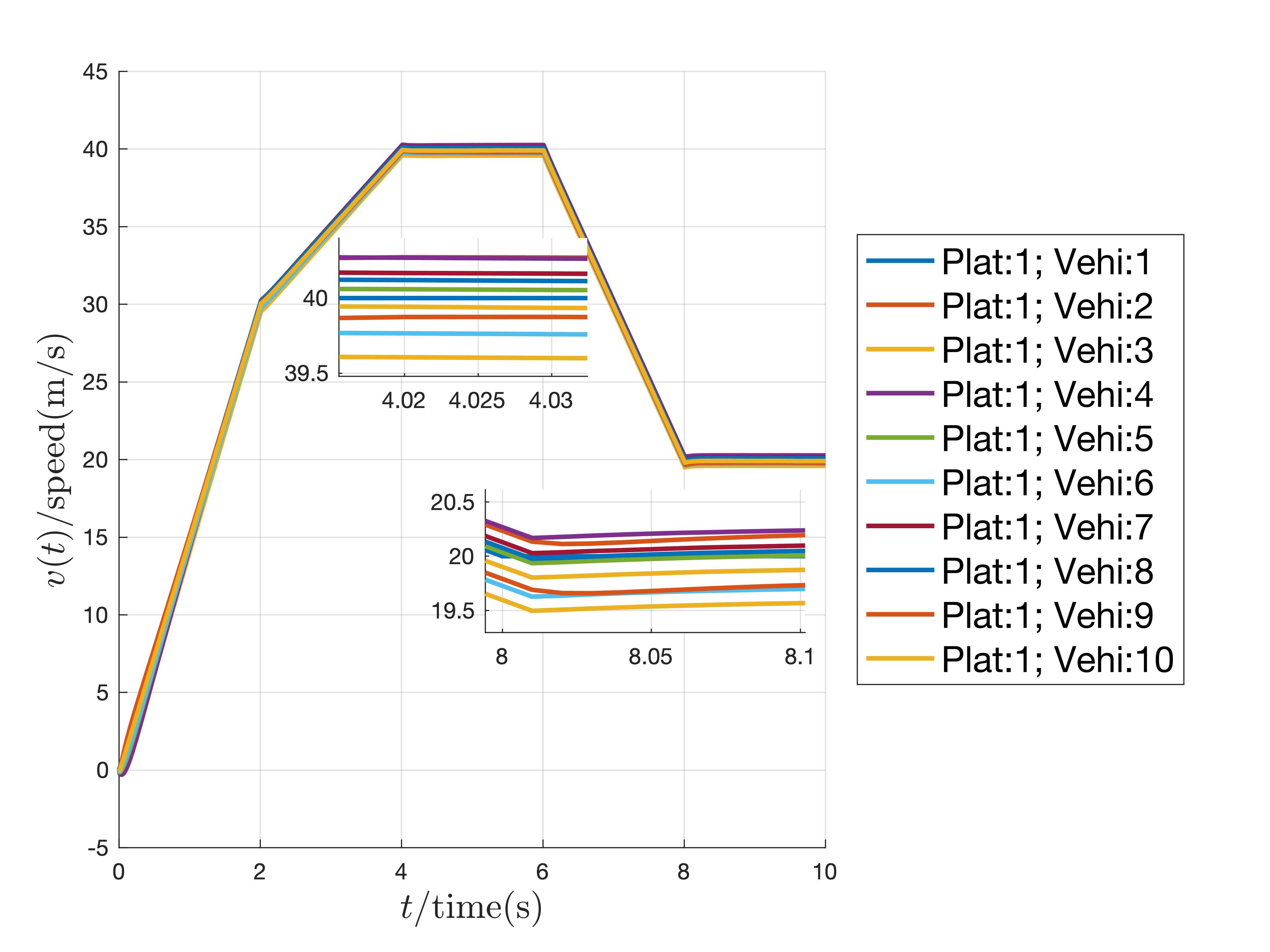}
    \vspace{-3mm}
    \caption{Velocity tracking of the platoon under centralized co-design.}
    \label{Fig:L2_WSS_VelocityTracking_CentralizedPlatoon}
\end{figure}

\begin{figure}[!h]
    \centering
    \includegraphics[width=3in]{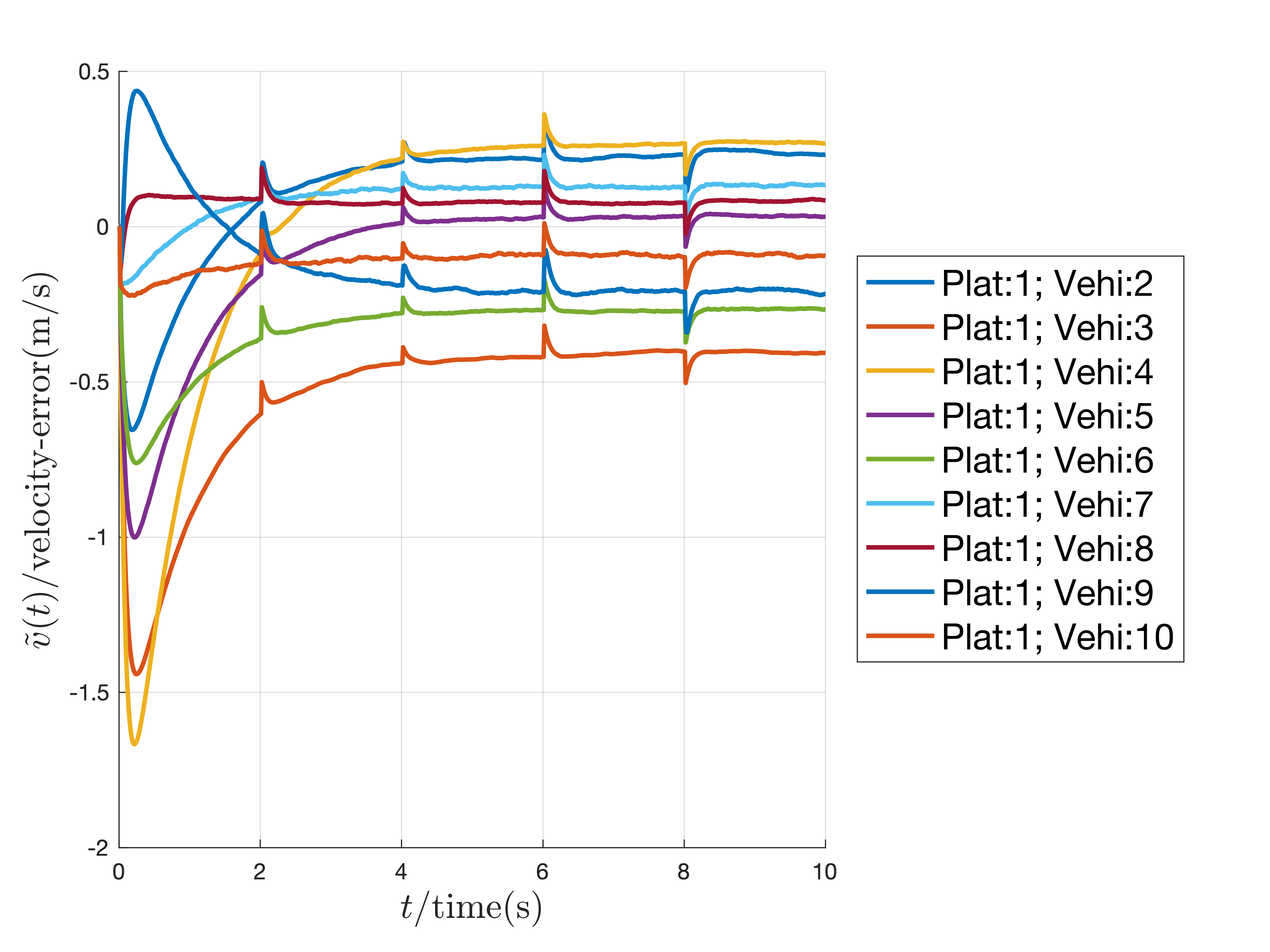}
    \vspace{-3mm}
    \caption{Velocity tracking errors under centralized co-design.}
    \label{Fig:L2_WSS_VelocityTrackingErrors_CentralizedPlatoon}
\end{figure}

For the vehicular platoon, we also verify the proposed decentralized co-design method. For the considered platoon, the synthesized decentralized communication topology is shown in Fig. \ref{Fig:L2_WSS_DecentralizedTopology}. Here, note that the follower vehicles are sequentially added/dropped in the simulation, i.e., the program \eqref{Eq:Th:DecentralizedTopologyDesign} in Theorem \ref{Th:DecentralizedTopologyDesign} is incrementally/iteratively solved.

An interesting phenomenon shown in Fig. \ref{Fig:L2_WSS_DecentralizedTopology} is the appearance/disappearance of communication links (controller gains $K_{ij}$) when vehicles joining/leaving the platoon. This practically makes sense because vehicular joining/leaving will correspondingly change the stability and dissipativity behaviors of the entire platoon, which also indicates the change of the positive definiteness of the matrices in \eqref{Eq:Th:DecentralizedTopologyDesign}.
Another phenomenon observed from Fig. \ref{Fig:L2_WSS_DecentralizedTopology} is the increased number of both leader-followers and followers-followers communication links as compared to the centralized case in Fig. \ref{Fig:L2_WSS_CentralizedTopology}. This is likely caused by the decentralized synthesis of the topology optimization problem \eqref{Eq:Th:DecentralizedTopologyDesign}, since the capability of the disturbance attenuation is incrementally adjusted, which ends up with a more conservative result than the centralized case in \eqref{Eq:Th:CentralizedTopologyDesign}. 

\begin{figure}[!h]
    \centering
    \includegraphics[width=3in]{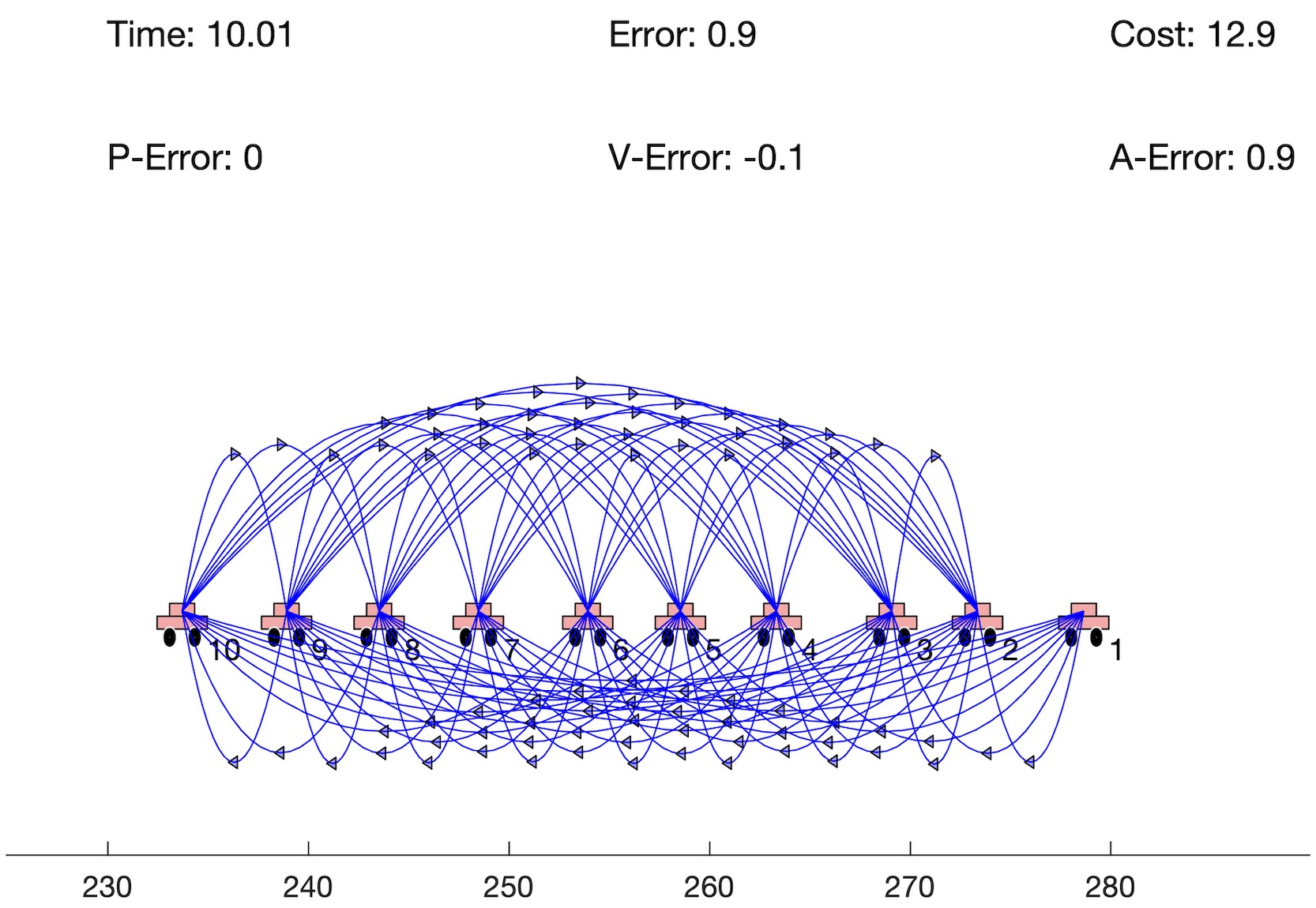}
    \vspace{-3mm}
    \caption{The platoon under the decentralized controller (topology). Directed arrows represent the interconnection topology, where the upper ones are the communication edges from the follower to the predecessor and vice versa. }
    \label{Fig:L2_WSS_DecentralizedTopology}
\end{figure}

\begin{figure}[!h]
    \centering
    \includegraphics[width=3in]{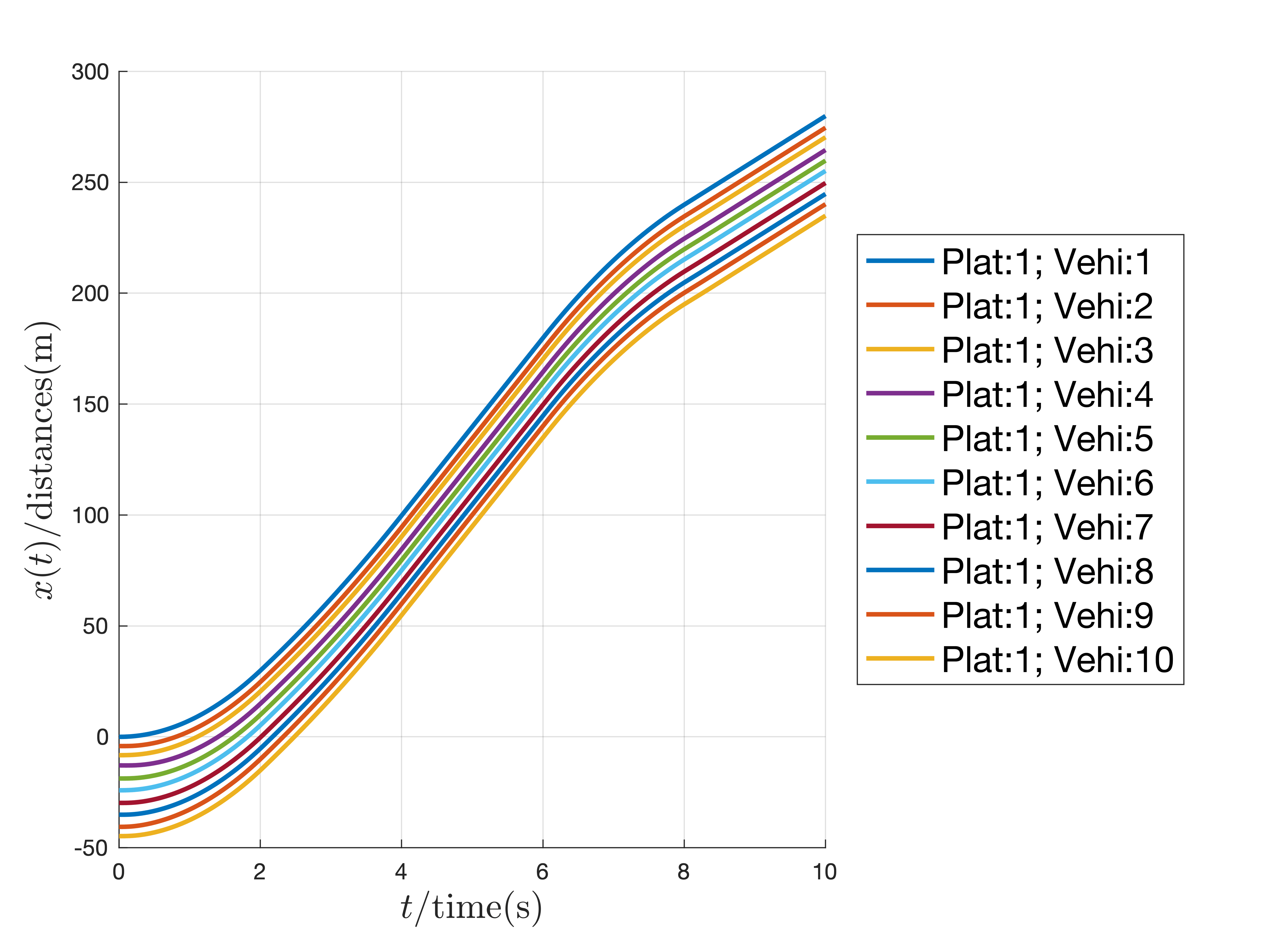}
    \vspace{-3mm}
    \caption{Position tracking under decentralized co-design.}
    \label{Fig:L2_WSS_PositionTracking_DecentralizedPlatoon}
\end{figure}

\begin{figure}[!h]
    \centering
    \includegraphics[width=3in]{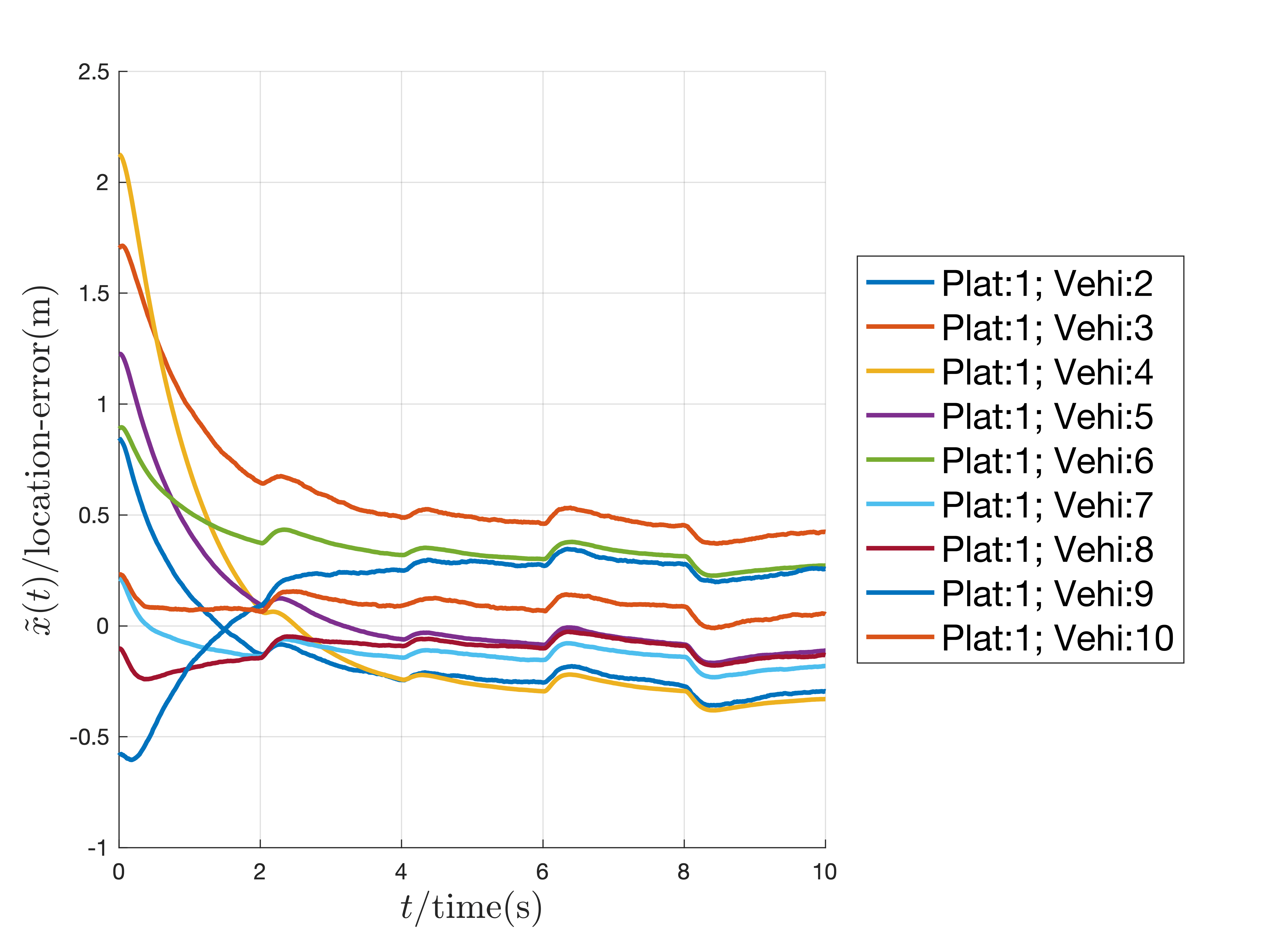}
    \caption{Position tracking errors under decentralized co-design.}
    \vspace{-3mm}
    \label{Fig:L2_WSS_PositionTrackingErrors_DecentralizedPlatoon}
\end{figure}

\begin{figure}[!h]
    \centering
    \includegraphics[width=3in]{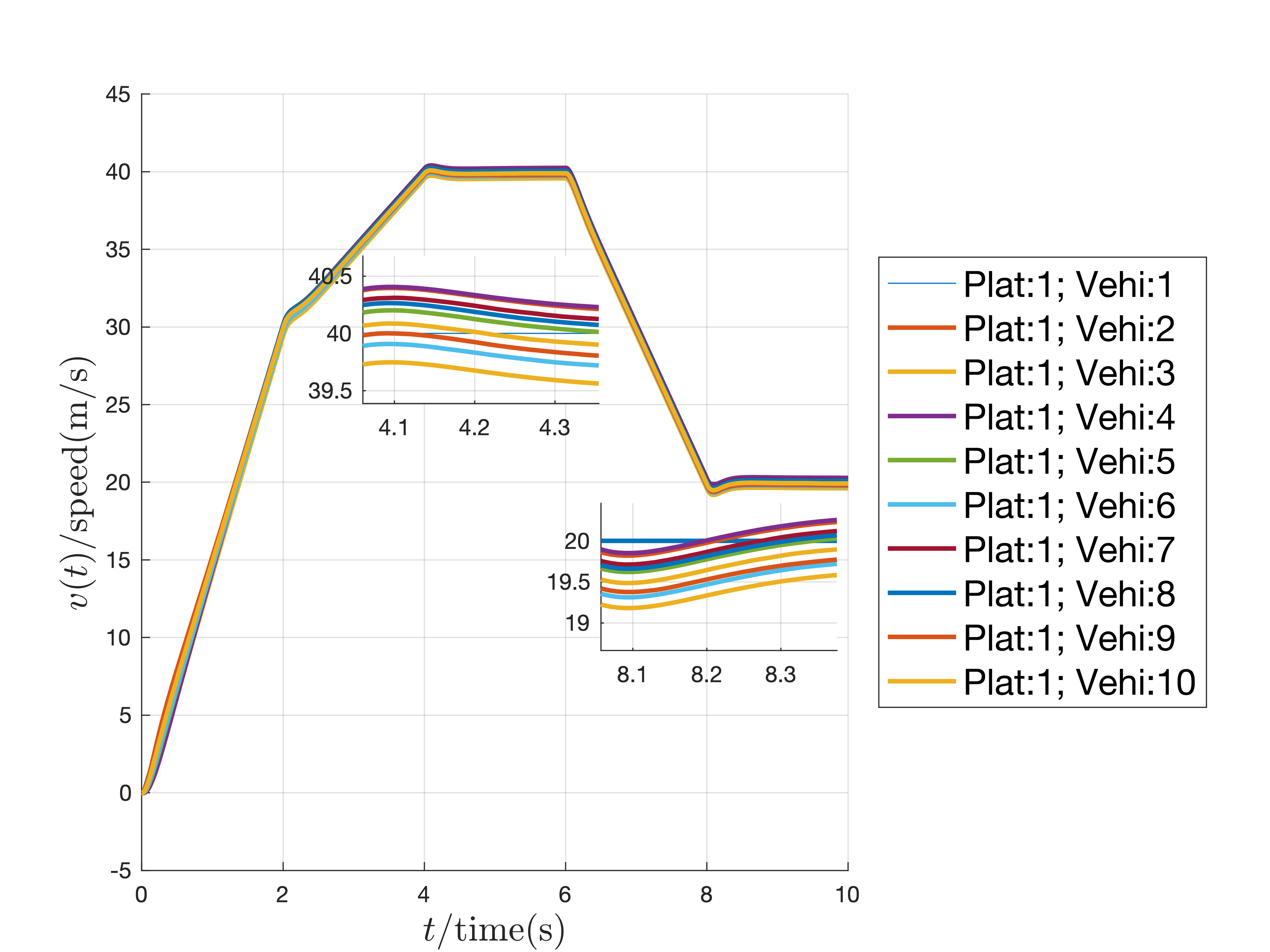}
    \vspace{-3mm}
    \caption{Velocity tracking under decentralized co-design.}
    \label{Fig:L2_WSS_PositionTrackingErrors_DecentralizedPlatoon}
\end{figure}

\begin{figure}[!h]
    \centering
    \includegraphics[width=3in]{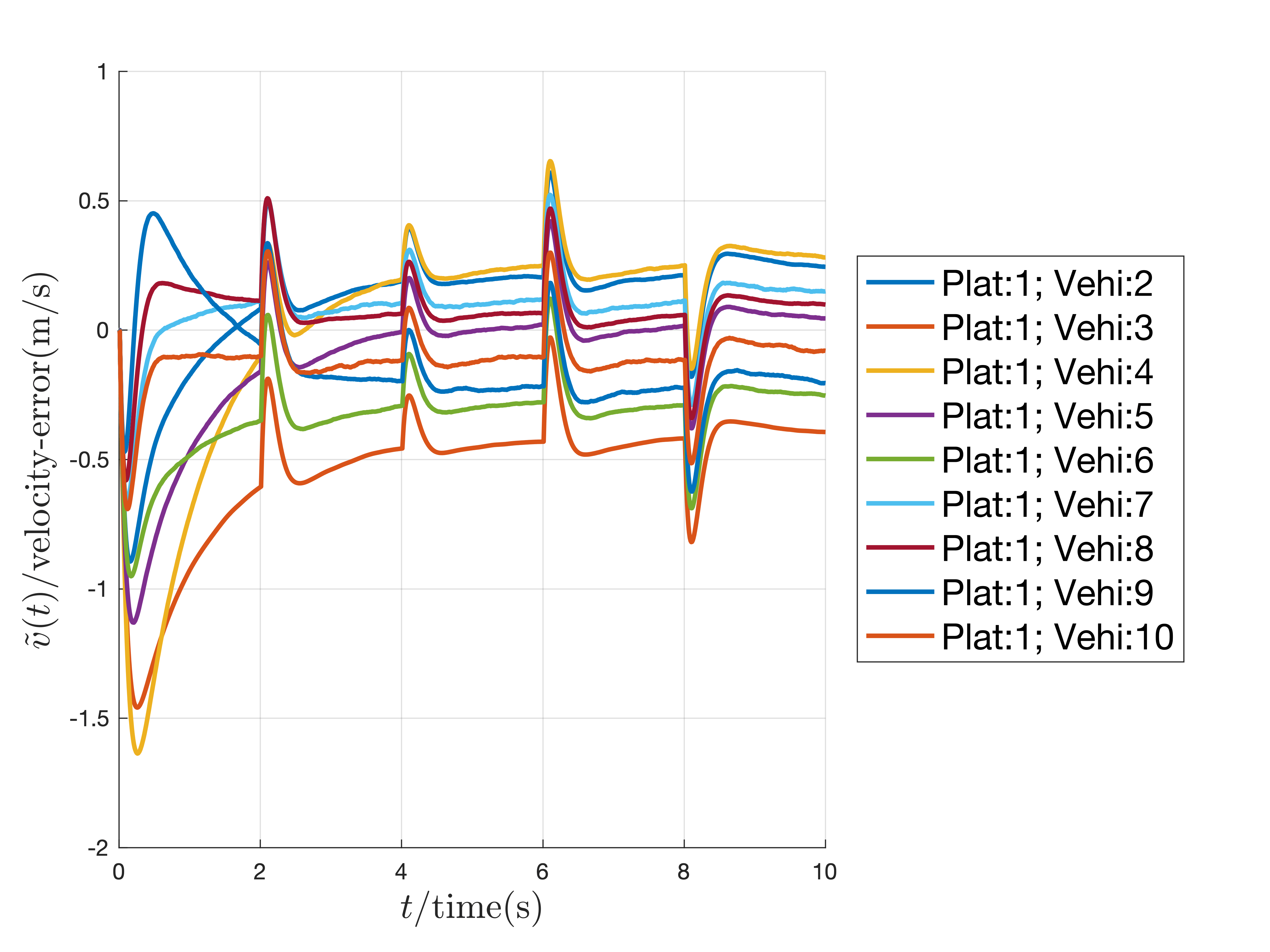}
    \vspace{-3mm}
    \caption{Velocity tracking errors under decentralized co-design.}
    \label{Fig:L2_WSS_VelocityTrackingErrors_DecentralizedPlatoon}
\end{figure}

The obtained (optimal) $L_2$-gain from the decentralized co-design is $5.2214$. This increased $L_2$-gain value (with respect to that in the centralized case) is indicative of the said conservatism (which led to more communication links) of the decentralized case.
To quantify the tracking control results and the string stability, the real-time position and velocity tracking results and the errors are shown in Figs. \ref{Fig:L2_WSS_PositionTracking_DecentralizedPlatoon}-\ref{Fig:L2_WSS_VelocityTrackingErrors_DecentralizedPlatoon}.
The major changes are the transient behaviors in the position and velocity tracking, but the string stability still holds from the observation that the position and velocity tracking errors are bounded and are non-increasing along the platoon. Similar to the centralized case, the weak coupling condition for DSS also holds, i.e., $\max\{\sum_{i\in\N_9}|P_iK_{ij}|\}=\max\{0.0139,0.0521,0.0165,0.0204,0.0081,0.0058,0.0165,\\
0.0128,0.0001\}<1$ holds for all $i\in\N_9$. Notice that these values are larger than those observed in the centralized case. This is intuitive as more communication links will increase the available avenues to propagate the errors and disturbances over the platoon.

\vspace{-1mm}

\section{Conclusion}\label{Sec:Conclusion}


In this paper, we proposed a dissipativity-based co-design framework for vehicular platoons' distributed controller and communication topology. 
This is a well-motivated problem and of immense practical importance, especially when a platoon is splitting into multiple platoons or two or more platoons merging. 
Our proposed design method is distributed since the controller for each vehicle only uses the information from the neighboring vehicles. 
Furthermore, our solution is compositional in that the distributed controller can be composed incrementally, and the controllers (topology weights) of the remaining vehicles do not need to be re-designed when vehicles join/leave the platoon. 
We showed that the proposed co-design method ensures the $L_2$ stability of the closed-loop platooning system, which further implies the $L_2$ weak string stability. 
Simulation results illustrated the applicability of the proposed method. 
In particular, it was observed that both centralized and decentralized co-design frameworks provide feasible solutions, but the decentralized case requires more communication links than that in the centralized case. The ongoing research explores establishing stronger string stability results and the impact of vehicle indexes and communication cost coefficients used in the decentralized co-design framework.

\section*{appendix}\label{Sec:Appendix}
\subsection{Proof of Proposition \ref{Pr:LTISystemXDisspativity}}
    
    \begin{proof}
Note that, under a constant input $u(t)=u^*$, the state $x(t)=x^*$ is an equilibrium state of \eqref{Eq:Pr:LTISystemXDisspativity1} if $Ax^*+Bu^*=\0$, and the corresponding equilibrium output is $y(t) = y^* \triangleq Cx^*+Du^*$. 

\textbf{(Necessity)} Consider the candidate storage function $V(x,x^*) = (x-x^*)^\T P (x-x^*)$ where $P>0$ which leads to 
\begin{align*}
    \dot{V}(x,x^*) = (x-x^*)^\T P (Ax + Bu) + (Ax + Bu)^\T P(x-x^*).
\end{align*}
Now, since $Ax+Bu=Ax^*+Bu^*$ where $\bar{x}\triangleq x-x^*$ and $\bar{u} \triangleq u - u^*$, we can simplify the above $\dot{V}(x,x^*)$ expression as 
\begin{equation}\label{Eq:Pr:LTISystemXDisspativity1Step1}
	\dot{V}(x,x^*) = \bar{x}^\T \H(PA)\bar{x} + \bar{x}^\T PB \bar{u} + \bar{u}^\T B^\T P \bar{x}.
\end{equation}

On the other hand, according to Def. \ref{Def:X-EID}, for \eqref{Eq:Pr:LTISystemXDisspativity1} to be $X$-EID, 
\begin{equation*}
    \dot{V}(x,x^*) \leq 
    \bm{u-u^*\\y-y^*}^\T 
    \bm{X^{11} & X^{12}\\X^{21} & X^{22}}
    \bm{u-u^*\\y-y^*}.
\end{equation*} 
However, since $y-y^* = C\bar{x}+D\bar{u}$, the above inequality can be simplified into the form: 
\begin{align}
	\dot{V}(x,x^*) 
	 \leq&\  
	 \bar{x}^\T (C^\T X^{22} C)\bar{x} 
	 + \bar{x}^\T (C^\T X^{21} + C^\T X^{22} D) \bar{u}\nonumber \\ 
	 &\ +\bar{u}^\T (X^{12}C + D^\T X^{22}C) \bar{x}\nonumber \\
	 &\ + \bar{u}^\T (X^{11} + \H(X^{12} D) + D^\T X^{22} D)\bar{u}.
	 \label{Eq:Pr:LTISystemXDisspativity1Step2}
\end{align} 

Combining \eqref{Eq:Pr:LTISystemXDisspativity1Step1} and \eqref{Eq:Pr:LTISystemXDisspativity1Step2}, we can obtain the condition:
\begin{equation*}
	\bm{\bar{x}\\\bar{u}}^\T W  \bm{\bar{x}\\\bar{u}} \geq 0,
\end{equation*}
where $W$ is the block matrix seen in the LMI in Prop. \ref{Pr:LTISystemXDisspativity}. This condition holds whenever $W \geq 0$, irrespective of $x,u,x^*$ and $u^*$. Thus, \eqref{Eq:Pr:LTISystemXDisspativity1} is $X$-EID if there exists $P>0$ such that $W \geq 0$.

\textbf{(Sufficiency)} Assume that we can find some $P>0$ such that $W\geq 0$ irrespective of $x,u,x^*$ and $u^*$. This implies that the storage function $V(x,x^*)\triangleq(x-x^*)^\T P (x-x^*)$ satisfies the conditions stated in Def. \ref{Def:EID} and  Def. \ref{Def:X-EID} (i.e., the conditions for \eqref{Eq:Pr:LTISystemXDisspativity1} to be $X$-EID, see also \cite{Madeira2022}). This completes the proof. 
\end{proof}

\subsection{Proof of Corollary \ref{Co:LTISystemXDisspativation}}

\begin{proof}
The proof starts by applying Prop. \ref{Pr:LTISystemXDisspativity} to the considered system \eqref{Eq:Co:LTISystemXDisspativation1}. Basically, in \eqref{Eq:Pr:LTISystemXDisspativity2}, we need to change the variables: $A \rightarrow A+BL$, $B \rightarrow \I$, $C\rightarrow\I$ and $D\rightarrow\0$, which gives
\begin{equation*}
\bm{-\mathcal{H}(P(A+BL))+ X^{22}& -P + X^{21}\\
\star & X^{11}}\geq 0.
\end{equation*}
Now, using the congruence principle (in particular, by pre- and post-multiplying by $\diag(P^{-1},\I)$), we can obtain
\begin{equation*}
\bm{-\mathcal{H}((A+BL)P^{-1})+ P^{-1}X^{22}P^{-1}& -I + P^{-1}X^{21}\\
\star & X^{11}}\geq 0.
\end{equation*}
Changing variables via $K \triangleq LP^{-1}$ and $P \triangleq P^{-1}$ then gives
\begin{equation*}
\bm{-\mathcal{H}(AP+BK)+ PX^{22}P& -I + PX^{21}\\
\star & X^{11}}\geq 0.
\end{equation*}
Finally, using $X^{22}<0$ and applying the Schur complements theory, we can obtain the  LMI condition in \eqref{Eq:Co:LTISystemXDisspativation2}.
\end{proof}

\subsection{Proof of Proposition \ref{Pr:NSC2Dissipativity}}

\begin{proof}
As each subsystem $\Sigma_i, i\in\N_N$ is $X_i$-EID, there exists a storage function $V_i(x_i,x_i^*)$ such that 
\begin{equation}\label{Eq:Pr:NSC2DissipativityStep1}
    \dot{V}_i(x_i,x_i^*) \leq \bm{u_i-u_i^*\\y_i-y_i^*}^\T X_i \bm{u_i-u_i^*\\y_i-y_i^*},
\end{equation}
for all $(x_i,x_i^*,u_i)\in \R^{n_i}\times \X_i \times \R^{q_i}$ and $i\in\N_N$. Now, for the networked system $\Sigma$, consider the candidate storage function $V(x,x^*)\triangleq \sum_{i\in\N_N} p_i V_i(x_i,x_i^*)$, which, with \eqref{Eq:Pr:NSC2DissipativityStep1}, leads to 
\begin{equation*}
\dot{V}(x,x^*) \leq  \bm{u-u^*\\y-y^*}^\T \textbf{X}_p \bm{u-u^*\\y-y^*},
\end{equation*}
for all $(x,x^*,u)\in \R^n\times \X \times \R^q$. Consequently, for the networked system $\Sigma$ to be $\textbf{Y}$-EID, we require 
\begin{equation}\label{Eq:Pr:NSC2DissipativityStep2}
\bm{u-u^*\\y-y^*}^\T \textbf{X}_p \bm{u-u^*\\y-y^*} \leq 
\bm{w-w^*\\z-z^*}^\T \textbf{Y} \bm{w-w^*\\z-z^*}.
\end{equation}
Besides, using the interconnection relationship \eqref{Eq:NSC2Interconnection} and equilibrium properties of the networked system $\Sigma$, we can write 
\begin{equation}\label{Eq:Pr:NSC2DissipativityStep3}
\begin{aligned}
    \bm{u-u^*\\y-y^*} =& \bm{M_{uy} & M_{uw} \\ \I & \0} \bm{y-y^*\\w-w^*},\\
    \bm{w-w^*\\z-z^*} =& \bm{\0 & \I \\ M_{zy} & M_{zw}}\bm{y-y^*\\w-w^*}.
\end{aligned}
\end{equation}
Finally, by applying \eqref{Eq:Pr:NSC2DissipativityStep3} in \eqref{Eq:Pr:NSC2DissipativityStep2} and re-arranging the terms, we can obtain the LMI in \eqref{Eq:Pr:NSC2Dissipativity1}. This completes the proof.
\end{proof}

\subsection{Proof of Proposition \ref{Pr:NSC2Synthesis}}

\begin{proof}
First, we restate the condition for the networked system $\Sigma$ to be $\textbf{Y}$-EID (i.e., \eqref{Eq:Pr:NSC2Dissipativity1} from Prop. \ref{Pr:NSC2Dissipativity}) as
\begin{equation}\label{Eq:Pr:NSC2SynthesisStep1}
    \begin{aligned}
    &\scriptsize
    \bm{M_{uy} & M_{uw}\\ M_{zy} & M_{zw}}^\T 
    \bm{\textbf{X}_p^{11} & \0 \\ \0 & -\textbf{Y}^{22}}
    \bm{M_{uy} & M_{uw}\\ M_{zy} & M_{zw}}\\
    &\scriptsize +\bm{ M_{uy}^\T \textbf{X}_p^{12} + \textbf{X}_p^{21}M_{uy} + \textbf{X}_p^{22} &  \textbf{X}_p^{21}M_{uw} - M_{zy}^\T \textbf{Y}^{21} \\ M_{uw}^\T \textbf{X}_p^{12} - \textbf{Y}^{12}M_{zy} & -(M_{zw}^\T \textbf{Y}^{21} +  \textbf{Y}^{12}M_{zw} + \textbf{Y}^{11})}
    \normalsize   
    \leq 0.     
    \end{aligned}
\end{equation}
Note that, under As. \ref{As:NegativeDissipativity}, when $X_i^{11}>0, \forall i\in\N_N$, we get $\scriptsize \bm{\textbf{X}_p^{11} & \0 \\ \0 & -\textbf{Y}^{22}}>0$. Therefore, \cite[Lm. 1]{WelikalaP52022} is applicable to obtain an equivalent condition to \eqref{Eq:Pr:NSC2SynthesisStep1}, which, under the change of variables $L_{uy} \triangleq \textbf{X}_p^{11}M_{uy}$ and $L_{uw} \triangleq \textbf{X}_p^{11}M_{uw}$, leads to \eqref{Eq:Pr:NSC2Synthesis2}. On the other hand, under As. \ref{As:NegativeDissipativity}, when $X_i^{11}<0, \forall i\in\N_N$, \cite[Lm. 2]{WelikalaP52022} is applicable to obtain an alternative equivalent condition to \eqref{Eq:Pr:NSC2SynthesisStep1} as \eqref{Eq:Pr:NSC2Synthesis2Alternative} (see also \cite[Rm. 14]{Welikala2022Ax3}). 
\end{proof}

\vspace{-1mm}
\bibliographystyle{IEEEtran}
\bibliography{References}

\begin{thebibliography}{10}

\bibitem{Antsaklis2006}
Panos~J. Antsaklis and Anthony~N. Michel.
\newblock {\em {Linear Systems}}.
\newblock Birkhauser, 2006.
\newblock \href {https://doi.org/10.1007/0-8176-4435-0}
  {\path{doi:10.1007/0-8176-4435-0}}.

\bibitem{Arcak2022}
Murat Arcak.
\newblock {Compositional Design and Verification of Large-Scale Systems Using
  Dissipativity Theory}.
\newblock {\em IEEE Control Systems Magazine}, 42(2):51--62, 2022.
\newblock \href {https://doi.org/10.1109/MCS.2021.3139721}
  {\path{doi:10.1109/MCS.2021.3139721}}.

\bibitem{besselink2017string}
Bart Besselink and Karl~H Johansson.
\newblock String stability and a delay-based spacing policy for vehicle
  platoons subject to disturbances.
\newblock {\em IEEE Transactions on Automatic Control}, 62(9):4376--4391, 2017.

\bibitem{Boyd1994}
Stephen Boyd, Laurent {El Ghaoui}, Eric Feron, and Venkataramanan Balakrishnan.
\newblock {\em {Linear Matrix Inequalities in System and Control Theory}}.
\newblock SIAM, 1994.
\newblock \href {https://doi.org/10.1137/1.9781611970777}
  {\path{doi:10.1137/1.9781611970777}}.

\bibitem{chen2018robust}
Na~Chen, Meng Wang, Tom Alkim, and Bart Van~Arem.
\newblock A robust longitudinal control strategy of platoons under model
  uncertainties and time delays.
\newblock {\em Journal of Advanced Transportation}, 2018, 2018.

\bibitem{chou2019backstepping}
Fang-Chieh Chou, Shu-Xia Tang, Xiao-Yun Lu, and Alexandre Bayen.
\newblock Backstepping-based time-gap regulation for platoons.
\newblock In {\em Proc. of American Control Conference}, pages 730--735, 2019.

\bibitem{dasgupta2017merging}
Soumya Dasgupta, Varunkumar Raghuraman, Apratim Choudhury, T~Nagacharan Teja,
  and Justin Dauwels.
\newblock {Merging and splitting maneuver of platoons by means of a novel PID
  controller}.
\newblock In {\em 2017 IEEE Symposium Series on Computational Intelligence
  (SSCI)}, pages 1--8. IEEE, 2017.

\bibitem{Madeira2022}
Diego {de S. Madeira}.
\newblock {Necessary and Sufficient Dissipativity-Based Conditions for Feedback
  Stabilization}.
\newblock {\em IEEE Trans. on Automatic Control}, 67(4):2100--2107, 2022.
\newblock \href {https://doi.org/10.1109/TAC.2021.3074850}
  {\path{doi:10.1109/TAC.2021.3074850}}.

\bibitem{fiengo2019distributed}
Giovanni Fiengo, Dario~Giuseppe Lui, Alberto Petrillo, Stefania Santini, and
  Manuela Tufo.
\newblock Distributed robust pid control for leader tracking in uncertain
  connected ground vehicles with v2v communication delay.
\newblock {\em IEEE/ASME Transactions on Mechatronics}, 24(3):1153--1165, 2019.

\bibitem{goli2019mpc}
Mohammad Goli and Azim Eskandarian.
\newblock {MPC-based lateral controller with look-ahead design for autonomous
  multi-vehicle merging into platoon}.
\newblock In {\em Proc. of American Control Conference}, pages 5284--5291,
  2019.

\bibitem{guo2015communication}
Ge~Guo and Shixi Wen.
\newblock {Communication scheduling and control of a platoon of vehicles in
  VANETs}.
\newblock {\em IEEE Transactions on intelligent transportation systems},
  17(6):1551--1563, 2015.

\bibitem{guo2016distributed}
Xianggui Guo, Jianliang Wang, Fang Liao, and Rodney Swee~Huat Teo.
\newblock Distributed adaptive integrated-sliding-mode controller synthesis for
  string stability of vehicle platoons.
\newblock {\em IEEE Transactions on Intelligent Transportation Systems},
  17(9):2419--2429, 2016.

\bibitem{hao2012achieving}
He~Hao and Prabir Barooah.
\newblock On achieving size-independent stability margin of vehicular lattice
  formations with distributed control.
\newblock {\em IEEE Transactions on Automatic Control}, 57(10):2688--2694,
  2012.

\bibitem{hao2011stability}
He~Hao, Prabir Barooah, and Prashant~G Mehta.
\newblock Stability margin scaling laws for distributed formation control as a
  function of network structure.
\newblock {\em IEEE Transactions on Automatic Control}, 56(4):923--929, 2011.

\bibitem{herman2014nonzero}
Ivo Herman, Dan Martinec, Zden{\v{e}}k Hur{\'a}k, and Michael {\v{S}}ebek.
\newblock Nonzero bound on fiedler eigenvalue causes exponential growth of
  h-infinity norm of vehicular platoon.
\newblock {\em IEEE Transactions on Automatic Control}, 60(8):2248--2253, 2014.

\bibitem{ji2018adaptive}
Xuewu Ji, Xiangkun He, Chen Lv, Yahui Liu, and Jian Wu.
\newblock Adaptive-neural-network-based robust lateral motion control for
  autonomous vehicle at driving limits.
\newblock {\em Control Engineering Practice}, 76:41--53, 2018.

\bibitem{jia2015survey}
Dongyao Jia, Kejie Lu, Jianping Wang, Xiang Zhang, and Xuemin Shen.
\newblock A survey on platoon-based vehicular cyber-physical systems.
\newblock {\em IEEE Communications Surveys \& Tutorials}, 18(1):263--284, 2015.

\bibitem{li2010design}
B~Li and F~Yu.
\newblock Design of a vehicle lateral stability control system via a fuzzy
  logic control approach.
\newblock {\em Proc. of the Institution of Mechanical Engineers, Part D:
  Journal of Automobile Engineering}, 224(3):313--326, 2010.

\bibitem{Lofberg2004}
J~Lofberg.
\newblock {YALMIP : A Toolbox for Modeling and Optimization in MATLAB}.
\newblock In {\em Proc. of IEEE Intl. Conf. on Robotics and Automation}, pages
  284--289, 2004.
\newblock \href {https://doi.org/10.1109/CACSD.2004.1393890}
  {\path{doi:10.1109/CACSD.2004.1393890}}.

\bibitem{mirabilio2021scalable}
Marco Mirabilio, Alessio Iovine, Elena De~Santis, Maria~Domenica Di~Benedetto,
  and Giordano Pola.
\newblock Scalable mesh stability of nonlinear interconnected systems.
\newblock {\em IEEE Control Systems Letters}, 6:968--973, 2021.

\bibitem{morales2016merging}
Am{\'e}rica Morales and Henk Nijmeijer.
\newblock Merging strategy for vehicles by applying cooperative tracking
  control.
\newblock {\em IEEE Transactions on Intelligent Transportation Systems},
  17(12):3423--3433, 2016.

\bibitem{morbidi2013decentralized}
Fabio Morbidi, Patrizio Colaneri, and Thomas Stanger.
\newblock Decentralized optimal control of a car platoon with guaranteed string
  stability.
\newblock In {\em European Control Conference}, pages 3494--3499, 2013.

\bibitem{orki2019control}
Omer Orki and Shai Arogeti.
\newblock Control of mixed platoons consist of automated and manual vehicles.
\newblock In {\em 2019 IEEE International Conference on Connected Vehicles and
  Expo (ICCVE)}, pages 1--6. IEEE, 2019.

\bibitem{ploeg2013controller}
Jeroen Ploeg, Dipan~P Shukla, Nathan Van De~Wouw, and Henk Nijmeijer.
\newblock Controller synthesis for string stability of vehicle platoons.
\newblock {\em IEEE Transactions on Intelligent Transportation Systems},
  15(2):854--865, 2013.

\bibitem{ploeg2013lp}
Jeroen Ploeg, Nathan Van De~Wouw, and Henk Nijmeijer.
\newblock Lp string stability of cascaded systems: Application to vehicle
  platooning.
\newblock {\em IEEE Transactions on Control Systems Technology},
  22(2):786--793, 2013.

\bibitem{Zihao2022Ax}
Zihao Song, Shirantha Welikala, Panos~J Antsaklis, and Hai Lin.
\newblock {Distributed Adaptive Backstepping Control for Vehicular Platoons
  with Mismatched Disturbances Using Vector String Lyapunov Functions}.
\newblock {\em arXiv e-prints}, page 2207.02101, 2022.
\newblock URL: \url{http://arxiv.org/abs/2207.02101}, \href
  {https://arxiv.org/abs/2207.02101} {\path{arXiv:2207.02101}}.

\bibitem{Zihao2022b}
Zihao Song, Shirantha Welikala, Panos~J. Antsaklis, and Hai Lin.
\newblock {Distributed Adaptive Backstepping Control for Vehicular Platoons
  with Mismatched Disturbances Using Vector String Lyapunov Functions}.
\newblock In {\em Proc. of American Control Conf.}, pages 1--6, 2023.

\bibitem{sontag1995characterizations}
Eduardo~D Sontag and Yuan Wang.
\newblock On characterizations of the input-to-state stability property.
\newblock {\em Systems \& Control Letters}, 24(5):351--359, 1995.

\bibitem{Sun2020}
Chuangchuang Sun, Shirantha Welikala, and Christos~G. Cassandras.
\newblock {Optimal Composition of Heterogeneous Multi-Agent Teams for Coverage
  Problems with Performance Bound Guarantees}.
\newblock {\em Automatica}, 117:108961, 2020.
\newblock \href {https://doi.org/10.1016/j.automatica.2020.108961}
  {\path{doi:10.1016/j.automatica.2020.108961}}.

\bibitem{syed2012coordinated}
Ali Syed, George Yin, Abhilash Pandya, Hongwei Zhang, et~al.
\newblock Coordinated vehicle platoon control: Weighted and constrained
  consensus and communication network topologies.
\newblock In {\em Proc. of 51st Conference on Decision and Control}, pages
  4057--4062, 2012.

\bibitem{tavan2015optimal}
Nematollah Tavan, Mehdi Tavan, and Rana Hosseini.
\newblock An optimal integrated longitudinal and lateral dynamic controller
  development for vehicle path tracking.
\newblock {\em Latin American Journal of Solids and Structures}, 12:1006--1023,
  2015.

\bibitem{Tibshirani1996}
Robert Tibshirani.
\newblock {Regression Shrinkage and Selection via the Lasso}.
\newblock {\em Journal of the Royal Statistical Society: Series B},
  58(1):267--288, 1996.

\bibitem{wang2022optimal}
Yan Wang, Rong Su, and Bohui Wang.
\newblock Optimal control of interconnected systems with time-correlated
  noises: Application to vehicle platoon.
\newblock {\em Automatica}, 137:110018, 2022.

\bibitem{WelikalaP32022}
Shirantha Welikala, Hai Lin, and Panos Antsaklis.
\newblock {A Generalized Distributed Analysis and Control Synthesis Approach
  for Networked Systems with Arbitrary Interconnections}.
\newblock In {\em Proc. of 30th Mediterranean Conf. on Control and Automation},
  pages 803--808, 2022.
\newblock \href {https://doi.org/10.1109/MED54222.2022.9837162}
  {\path{doi:10.1109/MED54222.2022.9837162}}.

\bibitem{WelikalaP42022}
Shirantha Welikala, Hai Lin, and Panos Antsaklis.
\newblock {On-line Estimation of Stability and Passivity Metrics}.
\newblock In {\em Proc. of 61st IEEE Conf. on Decision and Control (accepted)},
  2022.
\newblock URL: \url{http://arxiv.org/abs/2204.00073}.

\bibitem{Welikala2022Ax2}
Shirantha Welikala, Hai Lin, and Panos~J Antsaklis.
\newblock {A Generalized Distributed Analysis and Control Synthesis Approach
  for Networked Systems with Arbitrary Interconnections}.
\newblock {\em arXiv e-prints}, page 2204.09756, 2022.
\newblock URL: \url{http://arxiv.org/abs/2204.09756}, \href
  {https://arxiv.org/abs/2204.09756} {\path{arXiv:2204.09756}}.

\bibitem{Welikala2022Ax3}
Shirantha Welikala, Hai Lin, and Panos~J Antsaklis.
\newblock {Centralized and Decentralized Techniques for Analysis and Synthesis
  of Non-Linear Networked Systems}.
\newblock {\em arXiv e-prints}, page 2209.14552, 2022.
\newblock URL: \url{http://arxiv.org/abs/2209.14552}, \href
  {https://arxiv.org/abs/2209.14552} {\path{arXiv:2209.14552}}.

\bibitem{WelikalaP52022}
Shirantha Welikala, Hai Lin, and Panos~J. Antsaklis.
\newblock {Non-Linear Networked Systems Analysis and Synthesis using
  Dissipativity Theory}.
\newblock In {\em Proc. of American Control Conf.}, pages 2951--2956, 2023.
\newblock \href {https://doi.org/10.23919/ACC55779.2023.10155851}
  {\path{doi:10.23919/ACC55779.2023.10155851}}.

\bibitem{Willems1972a}
Jan~C. Willems.
\newblock {Dissipative Dynamical Systems Part I: General Theory}.
\newblock {\em Archive for Rational Mechanics and Analysis}, 45(5):321--351,
  1972.
\newblock \href {https://doi.org/10.1007/BF00276493}
  {\path{doi:10.1007/BF00276493}}.

\bibitem{xiao2011practical}
Lingyun Xiao and Feng Gao.
\newblock Practical string stability of platoon of adaptive cruise control
  vehicles.
\newblock {\em IEEE Transactions on Intelligent Transportation Systems},
  12(4):1184--1194, 2011.

\bibitem{yan2022pareto}
Yan Yan, Haiping Du, Defeng He, and Weihua Li.
\newblock Pareto optimal information flow topology for control of connected
  autonomous vehicles.
\newblock {\em IEEE Transactions on Intelligent Vehicles}, 8(1):330--343, 2022.

\bibitem{zheng2015stability}
Yang Zheng, Shengbo~Eben Li, Jianqiang Wang, Dongpu Cao, and Keqiang Li.
\newblock Stability and scalability of homogeneous vehicular platoon: Study on
  the influence of information flow topologies.
\newblock {\em IEEE Transactions on intelligent transportation systems},
  17(1):14--26, 2015.

\bibitem{zhu2018distributed}
Yang Zhu and Feng Zhu.
\newblock Distributed adaptive longitudinal control for uncertain third-order
  vehicle platoon in a networked environment.
\newblock {\em IEEE Transactions on Vehicular Technology}, 67(10):9183--9197,
  2018.

\end{thebibliography}

\end{document}